\documentclass[pra]{revtex4}

\usepackage{amssymb} \usepackage{graphicx}
\usepackage[utf8]{inputenc}
\usepackage[english]{babel} 
\usepackage{amsthm}

\newtheorem{definition}{Definition}
\newtheorem{proposition}{Proposition}
\newtheorem{corollary}{Corollary}
\newtheorem{theorem}{Theorem}

\begin{document}
 \title{Generalized symmetry superalgebras}

\author{\"{O}zg\"{u}r A\c{c}{\i}k}
\email{ozacik@science.ankara.edu.tr}
\address{Department of Physics,
Ankara University, Faculty of Sciences, 06100, Tando\u gan-Ankara,
Turkey\\}

\author{\"Umit Ertem}
 \email{umitertemm@gmail.com}
\address{Astronomer, Diyanet \.{I}\c{s}leri Ba\c{s}kanl{\i}\u{g}{\i}, \"{U}niversiteler Mah.\\
 Dumlup{\i}nar Bul. No:147/H 06800 \c{C}ankaya, Ankara, Turkey\\}

\date{\today}

\begin{abstract}

We generalize the symmetry superalgebras of isometries and geometric Killing spinors on a manifold to include all the hidden symmetries of the manifold generated by Killing spinors in all dimensions. We show that bilinears of geometric Killing spinors produce special Killing-Yano and special conformal Killing-Yano forms. After defining the Lie algebra structure of hidden symmetries generated by Killing spinors, we construct the symmetry operators as the generalizations of the Lie derivative on spinor fields. All these constructions together constitute the structure of generalized symmetry superalgebras. We exemplify the construction on weak $G_2$ and nearly K\"{a}hler manifolds.

\end{abstract}


\maketitle

\section{Introduction}

Killing vector fields generate the isometries of a manifold and they constitute a Lie algebra structure called the symmetry algebra of the manifold under the Lie bracket of vector fields. A manifold admitting a spin structure is called a spin manifold and the isometries of a spin manifold are related to the special types of spinors called geometric Killing spinors \cite{Lichnerowicz1,Alekseevsky Cortes1,Acik}. The squaring map of geometric Killing spinors correspond to Killing vector fields and they together form a superalgebra structure called the symmetry superalgebra on the manifold \cite{Klinker}. The even part of the superalgebra is the symmetry algebra of Killing vector fields and the odd part is the space of geometric Killing spinors \cite{OFarrill}. The brackets of the superalgebra correspond to the Lie bracket of vector fields for the even-even case, Lie derivative of spinor fields with respect to Killing vectors for the even-odd case and the squaring map of spinors for the odd-odd case. If all the brackets satisfy the Jacobi identities, then the symmetry superalgebra correspond to a Lie superalgebra. Besides the geometric Killing spinors, supergravity Killing spinors which are supersymmetry generators of bosonic supergravity theories can also be used in the construction of symmetry superalgebras called Killing superalgebras on supergravity backgrounds that are solutions of the supergravity field equations in various dimensions \cite{OFarrill HackettJones Moutsopoulos Simon,OFarrill Santi,deMedeiros OFarrill Santi}. These Killing superalgebras are important tools in the classification problem of supergravity backgrounds in all dimensions \cite{OFarrill Meessen Philip,OFarrill HackettJones Moutsopoulos,OFarrill Hustler1,OFarrill Hustler2}. They reduce to the symmetry superalgebras generated by geometric Killing spinors in constant curvature backgrounds since in that case supergravity Killing spinors reduce to geometric Killing spinors. In a flat background, the symmetry superalgebra is called the Poincare superalgebra and its subalgebra deformations can give way to find the Killing superalgebras of supergravity backgrounds \cite{OFarrill Santi1,OFarrill Santi2}.

Antisymmetric generalizations of isometries to higher degree differential forms are Killing-Yano (KY) forms which are called hidden symmetries of the manifold. They are related to the $p$-form Dirac currents of geometric Killing spinors which are generalizations of the squaring map of spinors to the higher degree bilinears \cite{Acik Ertem1}. The Poincare superalgebra can be extended to include these hidden symmetries in a consistent superalgebra structure \cite{Alekseevsky Cortes Devchand Proeyen,Alekseevsky Cortes2}. This comes from the fact that KY forms constitute a graded Lie algebra structure under the Schouten-Nijenhuis (SN) bracket of differential forms on constant curvature manifolds \cite{Kastor Ray Traschen}. Moreover, on constant curvature manifolds, the symmetry superalgebras can also be extended to include odd degree KY forms \cite{Ertem1,Ertem2}. The even part of the extended superalgebras correspond to the Lie algebra of odd degree KY forms under the SN bracket and the odd part is the space of geometric Killing spinors. The brackets of the superalgebra are the SN bracket for the even-even part, the symmetry operators of geometric Killing spinors which are generalizations of the Lie derivative on spinor fields on constant curvature manifolds for the even-odd part and the $p$-form Dirac currents of spinors for the odd-odd part. These extended superalgebras does not correspond to Lie superalgebras in general. However, these extensions cannot exhaust all the hidden symmetries constructed out of geometric Killing spinors and all backgrounds they can be defined.

In this paper, we generalize all the above symmetry superalgebras and their extensions to include all the hidden symmetries generated by geometric Killing spinors on manifolds they can be defined which exhaust all the possibilities. In that way, the algebraic structure of the hidden symmetries and geometric Killing spinors can be constructed for a given manifold admitting not only isometries but also the KY form generalizations of them. We show that the $p$-form Dirac currents of geometric Killing spinors correspond to special KY forms and special closed conformal KY (CCKY) forms which satisfy some special integrability conditions. The even or odd degree character of these forms are dependent on the spinor inner product defined on the manifold and the real or imaginary character of the geometric Killing spinor. We also prove that the CKY bracket defined in \cite{Ertem3} is the natural Lie algebra bracket for the special KY and special CCKY forms. We construct the symmetry operators of geometric Killing spinors which are relevant on all manifolds by using special KY and special CCKY forms. We finally show that all of these constructions together form superalgebra structures which correspond to the generalizations of symmetry superalgebras for all manifolds admitting geometric Killing spinors although they are not Lie superalgebras in general. The superalgebra structure is defined for all dimensions in the case that the even part of the superalgebra consists of special odd KY forms and special even CCKY forms. However, if the even part of the superalgebra is the Lie algebra of special even KY forms and special odd CCKY forms, then the generalized symmetry superalgebra is defined only in even dimensions. The brackets of the superalgebra structures in two different cases are slightly different from each other. We also demonstrate the structure of these generalized symmetry superalgebras on the examples of weak $G_2$ and nearly K\"{a}hler manifolds.

The paper is organized as follows. In Section 2, basic definitions and formulas are given and in Section 3, it is proved that the $p$-form Dirac currents of geometric Killing spinors correspond to special KY and special CCKY forms. The construction of the Lie algebra structure for the special KY and special CCKY forms is given in Section 4. The symmetry operators of geometric Killing spinors constructed out of special KY and special CCKY forms are the subject of Section 5. Section 6, includes the theorems proving the generalizations of symmetry superalgebras including all the hidden symmetries of a manifold. The examples of generalized symmetry superalgebras for weak $G_2$ and nearly K\"{a}hler manifolds are given in Section 7. Section 8 concludes the paper.

\section{Basic Definitions and Formulas}

We consider an $n$-dimensional spin manifold $M$, on which one can define a spin structure. Besides the exterior bundle $\Lambda M$, we can also construct Clifford bundle $Cl(M)$ and spinor bundle $\Sigma M$ on $M$. The sections of $Cl(M)$ corresponds to inhomogeneous differential forms and the sections of $\Sigma M$ are spinor fields.

\begin{definition}
For $\psi, \phi \in\Sigma M$, we can define the spin-invariant inner product $(\,,\,)$ with the property
\begin{equation}
(\psi,\phi)=\pm(\phi,\psi)^j
\end{equation}
where $j$ is an involution in the relevant Clifford algebra and the inner product takes values in the division algebra $\mathbb{D}$. So, $j$ can be identity (Id), complex conjugation ($^*$), quaternionic conjugation ($\bar{\,\,}$) or quaternionic reversion ($\widehat{\,\,}$) depending on the Clifford algebra corresponding to the matrix algebras over $\mathbb{D}=\mathbb{R}$, $\mathbb{C}$ or $\mathbb{H}$ in relevant dimensions \cite{Benn Tucker}. For $\psi,\phi\in\Sigma M$, $\alpha\in Cl(M)$ and $c\in\mathbb{D}$, the inner product has the following properties
\begin{eqnarray}
(\psi,\alpha.\phi)&=&(\alpha^{\mathcal{J}}.\psi,\phi)\\
(c\psi,\phi)&=&c^j(\psi,\phi)
\end{eqnarray}
where $.$ denotes the Clifford product and ${\mathcal{J}}$ is an involution on $Cl(M)$.
\end{definition}

\begin{definition}
We can define two different involutions on $Cl(M)$. The main automorphism $\eta:Cl(M)\rightarrow Cl(M)$ gives the $\mathbb{Z}_2$ grading to $Cl(M)=Cl^0(M)\oplus Cl^1(M)$ and has the property $\eta(\alpha.\beta)=\eta\alpha.\eta\beta$ for $\alpha,\beta\in Cl(M)$. For $\omega\in Cl(M)$, if $\eta\omega=\omega$, then $\omega\in Cl^0(M)$ and if $\eta\omega=-\omega$, then $\omega\in Cl^1(M)$. If, $\omega$ is a homogeneous $p$-form, then the action of $\eta$ on $\omega$ is defined by $\eta\omega=(-1)^p\omega$.

The involution $\xi:Cl(M)\rightarrow Cl(M)$ has the property $\xi(\alpha.\beta)=\xi\beta.\xi\alpha$ for $\alpha,\beta\in Cl(M)$. For a homogeneous $p$-form $\omega$, the action of $\xi$ is defined by $\xi\omega=(-1)^{\lfloor p/2 \rfloor}\omega$ where $\lfloor\rfloor$ denotes the floor function which takes the integer part of the argument. So, ${\mathcal{J}}$ in (2) can be $\xi$, $\xi\eta$, $\xi^*$ or $\xi\eta^*$ depending on the relevant Clifford algebra with $^*$ denoting the complex conjugation.
\end{definition}

The definition of the spinor inner product gives rise to the definition of the space of dual spinor fields $\Sigma^*M$. For a dual spinor $\overline{\psi}\in\Sigma^*M$, the action of it on a spinor $\phi\in\Sigma M$ is defined as \cite{Charlton}
\begin{equation}
\overline{\psi}(\phi)=(\psi,\phi).
\end{equation}
If we consider the tensor product of $\Sigma M\otimes\Sigma^*M$, its action on $\Sigma M$ is given by
\begin{equation}
(\psi\otimes\overline{\phi})\kappa=(\phi,\kappa)\psi
\end{equation}
for $\psi,\phi,\kappa\in\Sigma M$ and $\overline{\phi}\in\Sigma^*M$. This means that the elements of $\Sigma M\otimes\Sigma^*M$ correspond to the linear transformations of $\Sigma M$. Since $Cl(M)$ also acts on $\Sigma M$ via Clifford product corresponding to the linear transformations on $\Sigma M$, the tensor product $\Sigma M\otimes\Sigma^*M$ is isomorphic to $Cl(M)$. Then, the elements $\psi\otimes\overline{\phi}\in\Sigma M\otimes\Sigma^*M$ can be written as a sum of different degree differential forms corresponding to the elements of $Cl(M)$. For any orthonormal co-frame basis $\{e^a\}$, the Fierz identity which is the decomposition of the tensor product of spinors and dual spinors in terms of different degree differential forms is written as
\begin{eqnarray}
\psi\overline{\phi}&=&(\phi,\psi)+(\phi,e_a.\psi)e^a+(\phi,e_{ba}.\psi)e^{ab}+...+\nonumber\\
&&+(\phi,e_{a_p...a_2a_1}.\psi)e^{a_1a_2...a_p}+...+(-1)^{\lfloor n/2\rfloor}(\phi,z.\psi)z
\end{eqnarray}
where we denote $\psi\overline{\phi}=\psi\otimes\overline{\phi}$, $e^{a_1a_2...a_p}=e^{a_1}\wedge e^{a_2}\wedge...\wedge e^{a_p}$ and $z$ is the volume form.

\begin{definition}
For any $\psi\in\Sigma M$, the vector field $V_{\psi}$ which is the metric dual of the 1-form projection of $\psi\overline{\psi}$ that is
\begin{equation}
\widetilde{V}_{\psi}=(\psi\overline{\psi})_1=(\psi,e_a.\psi)e^a
\end{equation}
is called as the Dirac current of $\psi$ where $\,\widetilde{\,\, }$ denotes the metric dual.
\end{definition}
As a generalization of this definition, we also have
\begin{definition}
For any $\psi\in\Sigma M$, the $p$-form projection of $\psi\overline{\psi}$ that is
\begin{equation}
(\psi\overline{\psi})_p=(\psi,e_{a_p...a_2a_1}.\psi)e^{a_1a_2...a_p}
\end{equation}
is called as the $p$-form Dirac current of $\psi$.
\end{definition}

Two first-order differential operators can be defined on $\Sigma M$ via the Levi-Civita connection $\nabla$ induced on $\Sigma M$. The first one is the Dirac operator which is defined as
\begin{equation}
\displaystyle{\not}D=e^a.\nabla_{X_a}
\end{equation}
for the frame basis $\{X_a\}$ and the co-frame basis $\{e^a\}$. The second one is the twistor operator which is written as
\begin{equation}
\nabla_{X_a}-\frac{1}{n}e_a.\displaystyle{\not}D
\end{equation}
in $n$ dimensions. The spinors that are in kernel of the Dirac operator are called harmonic spinors and that are in the kernel of the twistor operator are called twistor spinors \cite{Lichnerowicz2,Habermann,Baum Leitner,Baum Friedrich Grunewald Kath,Bourguignon et al}.

\begin{definition}
A spinor field $\psi\in\Sigma M$, which is an eigenspinor of the Dirac operator $\displaystyle{\not}D\psi=m\psi$ and also a twistor spinor at the same time, satisfies the following Killing spinor equation
\begin{equation}
\nabla_{X_a}\psi=\lambda e_a.\psi
\end{equation}
and is called a geometric Killing spinor with the Killing number $\lambda:=\frac{m}{n}$. The Killing number $\lambda$ is a real or pure imaginary number.
\end{definition}
We call geometric Killing spinors as Killing spinors for simplicity. The existence of Killing spinors on a manifold $M$ constrains the curvature characteristics of the manifold. Integrability condition of the Killing spinor equation (11) is written as
\begin{equation}
R_{ab}.\psi=-4\lambda^2(e_a\wedge e_b).\psi
\end{equation}
where $R_{ab}$ are curvature 2-forms. Moreover, this implies
\begin{equation}
{\mathcal{R}}=-4\lambda^2n(n-1)
\end{equation}
where ${\mathcal{R}}$ is the scalar curvature.

The Dirac current $V_{\psi}$ of a Killing spinor $\psi$ corresponds to a Killing vector field, namely its metric dual satisfies the following equality
\begin{equation}
\nabla_X\widetilde{V}_{\psi}=\frac{1}{2}i_Xd\widetilde{V}_{\psi}
\end{equation}
for any vector field $X$ where $d$ is the exterior derivative operator and $i_X$ is the interior derivative (contraction) operator with respect to $X$.

Killing vector fields have antisymmetric generalizations to higher-degree differential forms.
\begin{definition}
If a $p$-form $\omega$ satisfies
\begin{equation}
\nabla_X\omega=\frac{1}{p+1}i_Xd\omega
\end{equation}
for any vector field $X$, then $\omega$ is called as a Killing-Yano (KY) $p$-form.
\end{definition}
KY forms are special cases of conformal Killing-Yano (CKY) forms.
\begin{definition}
If a $p$-form $\omega$ satisfies
\begin{equation}
\nabla_X\omega=\frac{1}{p+1}i_Xd\omega-\frac{1}{n-p+1}\widetilde{X}\wedge\delta\omega
\end{equation}
for any vector field $X$, its metric dual $\widetilde{X}$ and the co-derivative operator $\delta$, then $\omega$ is called as a CKY $p$-form.
\end{definition}
As can be seen from the last definition that the co-closed CKY forms satisfying $\delta\omega=0$ correspond to KY forms. On the other hand, another subset of CKY forms consists of closed CKY (CCKY) forms satisfying $d\omega=0$, namely the following equation
\begin{equation}
\nabla_X\omega=-\frac{1}{n-p+1}\widetilde{X}\wedge\delta\omega.
\end{equation}
CKY equation (16) has Hodge duality invariance, namely if $\omega$ is a CKY $p$-form then $*\omega$ is also a CKY $(n-p)$-form with $*$ denoting the Hodge map \cite{Semmelmann}. One can show that this property leads to the fact that for a KY $p$-form $\omega$, its Hodge dual $*\omega$ is a CCKY $(n-p)$-form and conversely for a CCKY $p$-form $\omega$, its Hodge dual $*\omega$ is a KY $(n-p)$-form. So, KY $p$-forms and CCKY $(n-p)$-forms are Hodge duals of each other \cite{Ertem4}.

The integrability condition of the KY equation (15) can be found as follows \cite{Semmelmann,Acik Ertem Onder Vercin}
\begin{equation}
\nabla_{X_a}d\omega=-\frac{p+1}{p}R_{ab}\wedge i_{X^b}\omega
\end{equation}
and similarly for the CCKY equation (17), the integrability condition is \cite{Ertem3}
\begin{equation} 
\nabla_{X_a}\delta\omega=-\frac{n-p+1}{n-p}i_{X^c}\left(i_{X_b}R_{ca}\wedge i_{X^b}\omega\right).
\end{equation}
We can define subsets of KY and CCKY forms which satisfy special types of integrability conditions.
\begin{definition}
A KY $p$-form $\omega$ is called a special KY $p$-form, if it satisfies the following condition
\begin{equation}
\nabla_X d\omega=-c(p+1)\widetilde{X}\wedge\omega
\end{equation}
where $c$ is a constant. Similarly, a CCKY $p$-form is called a special CCKY $p$-form if it satisfies
\begin{equation}
\nabla_{X}\delta\omega=c(n-p+1)i_X\omega
\end{equation}
where $c$ is a constant.
\end{definition}
Note that, in constant curvature manifolds with curvature 2-forms $R_{ab}=ce_a\wedge e_b$ for a constant $c$, the integrability conditions (18) and (19) reduces to (20) and (21), respectively. This means that in constant curvature manifolds, all KY forms are special KY forms and all CCKY forms are special CCKY forms. Obviously, this is not true in general.

\section{Bilinears of Killing Spinors}

Besides the relation between Dirac currents of Killing spinors and Killing vectors, $p$-form Dirac currents of Killing spinors are also related to KY and CCKY forms. For a Killing spinor $\psi$, its $p$-form Dirac current $(\psi\overline{\psi})_p$ is a KY $p$-form or a CCKY $p$-form depending on the chosen involution of the inner product, on the real or imaginary character of $\lambda$ and on the parity (evenness or oddness) of $p$ \cite{Acik Ertem1}. Moreover, we have the following
\begin{proposition}
For two Killing spinors $\psi$ and $\phi$ satisfying (11), the symmetric combination of $p$-form bilinears $(\psi\overline{\phi})_p+(\phi\overline{\psi})_p$ corresponds to a KY $p$-form or a CCKY $p$-form depending on the chosen involution of the inner product, on the real or imaginary character of $\lambda$ and on the parity (evenness or oddness) of $p$.
\end{proposition}
\begin{proof}
From the compatibility of the Levi-Civita connection $\nabla$ with the spinor inner product and the spinor duality operation, we can write the covariant derivative of $(\psi\overline{\phi})_p+(\phi\overline{\psi})_p$ with respect to any frame basis vector $X_a$ as
\begin{eqnarray}
\nabla_{X_a}\left[(\psi\overline{\phi})_p+(\phi\overline{\psi})_p\right]&=&\left((\nabla_{X_a}\psi)\overline{\phi}\right)_p+\left(\psi(\overline{\nabla_{X_a}\phi})\right)_p\nonumber\\
&&+\left((\nabla_{X_a}\phi)\overline{\psi}\right)_p+\left(\phi(\overline{\nabla_{X_a}\psi})\right)_p\nonumber\\
&=&\left(\lambda e_a.\psi\overline{\phi}\right)_p+\left(\psi(\overline{\lambda e_a.\phi})\right)_p\nonumber\\
&&+\left(\lambda e_a.\phi\overline{\psi}\right)_p+\left(\phi(\overline{\lambda e_a.\psi})\right)_p\nonumber
\end{eqnarray}
where we have used (11) in the second line. From the properties (5), (3) and (2), we can write the equality $\psi(\overline{\lambda e_a.\phi})=\lambda^j(\psi\overline{\phi}).e_a^{\mathcal{J}}$ and similarly $\phi(\overline{\lambda e_a.\psi})=\lambda^j(\phi\overline{\psi}).e_a^{\mathcal{J}}$. So, we have
\begin{eqnarray}
\nabla_{X_a}\left[(\psi\overline{\phi})_p+(\phi\overline{\psi})_p\right]&=&\left(\lambda e_a.\psi\overline{\phi}\right)_p+\left(\lambda^j(\psi\overline{\phi}).e_a^{\mathcal{J}}\right)_p\nonumber\\
&&+\left(\lambda e_a.\phi\overline{\psi}\right)_p+\left(\lambda^j(\phi\overline{\psi}).e_a^{\mathcal{J}}\right)_p.\nonumber
\end{eqnarray}
The Clifford product can be expressed in terms of wedge product and interior derivative. For a 1-form $x$ and an arbitrary form $\alpha$, we have the following identities
\begin{eqnarray}
x.\alpha&=&x\wedge\alpha+i_{\widetilde{x}}\alpha\\
\alpha.x&=&x\wedge\eta\alpha-i_{\widetilde{x}}\eta\alpha
\end{eqnarray}
where $\widetilde{x}$ denotes the vector field metric dual to $x$. By using these equalities, we obtain
\begin{eqnarray}
\nabla_{X_a}\left[(\psi\overline{\phi})_p+(\phi\overline{\psi})_p\right]&=&\lambda e_a\wedge\left(\psi\overline{\phi}\right)_{p-1}+\lambda i_{X_a}\left(\psi\overline{\phi}\right)_{p+1}\nonumber\\
&&+\lambda^je_a^{\mathcal{J}}\wedge\left(\psi\overline{\phi}\right)_{p+1}^{\eta}-\lambda^j i_{\widetilde{e_a^{\mathcal{J}}}}\left(\psi\overline{\phi}\right)_{p+1}^{\eta}\nonumber\\
&&+\lambda e_a\wedge\left(\phi\overline{\psi}\right)_{p-1}+\lambda i_{X_a}\left(\phi\overline{\psi}\right)_{p+1}\nonumber\\
&&+\lambda^je_a^{\mathcal{J}}\wedge\left(\phi\overline{\psi}\right)_{p+1}^{\eta}-\lambda^j i_{\widetilde{e_a^{\mathcal{J}}}}\left(\phi\overline{\psi}\right)_{p+1}^{\eta}
\end{eqnarray}
where we denote $\eta\left(\phi\overline{\psi}\right)_{p+1}=\left(\phi\overline{\psi}\right)_{p+1}^{\eta}$. For the torsion-free Levi-Civita connection $\nabla$, we have the identity $d=e^a\wedge\nabla_{X_a}$ and taking the wedge product of (24) with $e^a$ from the left gives
\begin{eqnarray}
d\left[(\psi\overline{\phi})_p+(\phi\overline{\psi})_p\right]&=&\lambda (p+1)\left(\psi\overline{\phi}\right)_{p+1}-\lambda^j\textrm{sgn}(e_a^{\mathcal{J}})(p+1)\left(\psi\overline{\phi}\right)_{p+1}^{\eta}\nonumber\\
&&+\lambda (p+1)\left(\phi\overline{\psi}\right)_{p+1}-\lambda^j\textrm{sgn}(e_a^{\mathcal{J}})(p+1)\left(\phi\overline{\psi}\right)_{p+1}^{\eta}\nonumber\\
\end{eqnarray}
where we have used $e^a\wedge e_a=0$ and $e^a\wedge i_{X_a}\alpha=p\alpha$ for a $p$-form $\alpha$. $\textrm{sgn}(e_a^{\mathcal{J}})$ is equal to $\pm1$ depending on the involution ${\mathcal{J}}$. Similarly, we have the identity $\delta=-i_{X^a}\nabla_{X_a}$ for the Levi-Civita connection $\nabla$ and by taking the interior derivative of (24) with respect to $i_{X^a}$, we obtain
\begin{eqnarray}
\delta\left[(\psi\overline{\phi})_p+(\phi\overline{\psi})_p\right]&=&-\lambda(n-p+1)\left(\psi\overline{\phi}\right)_{p-1}^{\eta}\nonumber\\
&&-\lambda^j\textrm{sgn}(e_a^{\mathcal{J}})(n-p+1)\left(\psi\overline{\phi}\right)_{p-1}^{\eta}\nonumber\\
&&-\lambda(n-p+1)\left(\phi\overline{\psi}\right)_{p-1}^{\eta}\nonumber\\
&&-\lambda^j\textrm{sgn}(e_a^{\mathcal{J}})(n-p+1)\left(\phi\overline{\psi}\right)_{p-1}^{\eta}\nonumber\\
\end{eqnarray}
where we have used $i_{X^a}e_a=n$ and  $e^a\wedge i_{X_a}\alpha=p\alpha$ for a $p$-form $\alpha$. In equalities (24), (25) and (26), we have four parameters to choose; $\lambda$ can be real (Re) or pure imaginary (Im), $j$ can be Id, $^*$, $\bar{\,\,}$ or $\widehat{\,\,}$, ${\mathcal{J}}$ can be $\xi$, $\xi^*$, $\xi\eta$ or $\xi\eta^*$ and $p$ can be even or odd. If $\lambda$ is real, we have $\lambda^{\textrm{Id}}=\lambda^*=\bar{\lambda}=\widehat{\lambda}=\lambda$ and if $\lambda$ is pure imaginary, we have $\lambda^{\textrm{Id}}=\widehat{\lambda}=\lambda$ and $\lambda^*=\bar{\lambda}=-\lambda$. Since $e_a$ is a real 1-form, we have $e_a^{\xi}=e_a^{\xi^*}=e_a$ with $\textrm{sgn}(e_a^{\mathcal{J}})=1$ and $e_a^{\xi\eta}=e_a^{\xi\eta^*}=-e_a$ with $\textrm{sgn}(e_a^{\mathcal{J}})=-1$. For $p$ even, we have $(\psi\overline{\phi})_{p-1}^{\eta}=-(\psi\overline{\phi})_{p-1}$, $(\phi\overline{\psi})_{p-1}^{\eta}=-(\phi\overline{\psi})_{p-1}$, $(\psi\overline{\phi})_{p+1}^{\eta}=-(\psi\overline{\phi})_{p+1}$, $(\phi\overline{\psi})_{p+1}^{\eta}=-(\phi\overline{\psi})_{p+1}$ and for $p$ odd, we have $(\psi\overline{\phi})_{p-1}^{\eta}=(\psi\overline{\phi})_{p-1}$, $(\phi\overline{\psi})_{p-1}^{\eta}=(\phi\overline{\psi})_{p-1}$, $(\psi\overline{\phi})_{p+1}^{\eta}=(\psi\overline{\phi})_{p+1}$, $(\phi\overline{\psi})_{p+1}^{\eta}=(\phi\overline{\psi})_{p+1}$. So, by considering these possibilities, we can obtain two different sets of equations from (24), (25) and (26). If the parameters ${\mathcal{J}}$, $\lambda$, $j$ and $p$ are chosen as in Table I, then we find the following set of equations from (24), (25) and (26)
\begin{eqnarray}
\nabla_{X_a}\left[(\psi\overline{\phi})_p+(\phi\overline{\psi})_p\right]&=&2\lambda i_{X_a}\left[(\psi\overline{\phi})_{p+1}+(\phi\overline{\psi})_{p+1}\right]\nonumber\\
d\left[(\psi\overline{\phi})_p+(\phi\overline{\psi})_p\right]&=&2\lambda(p+1)\left[(\psi\overline{\phi})_{p+1}+(\phi\overline{\psi})_{p+1}\right]\\
\delta\left[(\psi\overline{\phi})_p+(\phi\overline{\psi})_p\right]&=&0.\nonumber
\end{eqnarray}

\begin{table}[h]
{\centering{
\begin{tabular}{c c c c}

${\mathcal{J}}\quad$ & \quad $\lambda$\quad & \quad $j$\quad & \quad $p$ \\ \hline
$\xi,\xi^*\quad$ & \quad Re \quad & \quad $\textrm{Id}, ^*, \bar{\,\,}, \widehat{\,\,}$ \quad & \quad even \\
$\xi,\xi^*\quad$ & \quad Im \quad & \quad $\textrm{Id}, \widehat{\,\,}$ \quad & \quad even \\
$\xi,\xi^*\quad$ & \quad Im \quad & \quad $^*, \bar{\,\,}$ \quad & \quad odd \\
$\xi\eta,\xi\eta^*\quad$ & \quad Re \quad & \quad $\textrm{Id}, ^*, \bar{\,\,}, \widehat{\,\,}$ \quad & \quad odd \\
$\xi\eta,\xi\eta^*\quad$ & \quad Im \quad & \quad $\textrm{Id}, \widehat{\,\,}$ \quad & \quad odd \\
$\xi\eta,\xi\eta^*\quad$ & \quad Im \quad & \quad $^*, \bar{\,\,}$ \quad & \quad even \\

\end{tabular}}
\quad\\}
\caption{First set of possibilities for the parameters ${\mathcal{J}}$, $\lambda$, $j$ and $p$.}
\end{table}
We call this set of equations as Case 1 and by comparing them with each other, one can easily see that $(\psi\overline{\phi})_p+(\phi\overline{\psi})_p$ satisfies the following equation
\begin{equation}
\nabla_{X_a}\left[(\psi\overline{\phi})_p+(\phi\overline{\psi})_p\right]=\frac{1}{p+1}i_{X_a}d\left[(\psi\overline{\phi})_p+(\phi\overline{\psi})_p\right]
\end{equation}
and hence from (15) it is a KY $p$-form.

For the parameters ${\mathcal{J}}$, $\lambda$, $j$ and $p$, the remaining possibilities are given as in Table II.

\begin{table}[h]
{\centering{
\begin{tabular}{c c c c}

${\mathcal{J}}\quad$ & \quad $\lambda$\quad & \quad $j$\quad & \quad $p$ \\ \hline
$\xi,\xi^*\quad$ & \quad Re \quad & \quad $\textrm{Id}, ^*, \bar{\,\,}, \widehat{\,\,}$ \quad & \quad odd \\
$\xi,\xi^*\quad$ & \quad Im \quad & \quad $\textrm{Id}, \widehat{\,\,}$ \quad & \quad odd \\
$\xi,\xi^*\quad$ & \quad Im \quad & \quad $^*, \bar{\,\,}$ \quad & \quad even \\
$\xi\eta,\xi\eta^*\quad$ & \quad Re \quad & \quad $\textrm{Id}, ^*, \bar{\,\,}, \widehat{\,\,}$ \quad & \quad even \\
$\xi\eta,\xi\eta^*\quad$ & \quad Im \quad & \quad $\textrm{Id}, \widehat{\,\,}$ \quad & \quad even \\
$\xi\eta,\xi\eta^*\quad$ & \quad Im \quad & \quad $^*, \bar{\,\,}$ \quad & \quad odd \\

\end{tabular}}
\quad\\}
\caption{Second set of possibilities for the parameters ${\mathcal{J}}$, $\lambda$, $j$ and $p$.}
\end{table}
Note that the parameter sets in this case differ from the first case only in the parity of $p$. By considering these possibilities, we find the following set of equation from (24), (25) and (26)
\begin{eqnarray}
\nabla_{X_a}\left[(\psi\overline{\phi})_p+(\phi\overline{\psi})_p\right]&=&2\lambda e_a\wedge\left[(\psi\overline{\phi})_{p-1}+(\phi\overline{\psi})_{p-1}\right]\nonumber\\
d\left[(\psi\overline{\phi})_p+(\phi\overline{\psi})_p\right]&=&0\\
\delta\left[(\psi\overline{\phi})_p+(\phi\overline{\psi})_p\right]&=&-2\lambda(n-p+1)\left[(\psi\overline{\phi})_{p-1}+(\phi\overline{\psi})_{p-1}\right].\nonumber
\end{eqnarray}
We call this set of equations as Case 2 and by comparing them with each other, one finds that $(\psi\overline{\phi})_p+(\phi\overline{\psi})_p$ satisfies the following equation
\begin{equation}
\nabla_{X_a}\left[(\psi\overline{\phi})_p+(\phi\overline{\psi})_p\right]=-\frac{1}{n-p+1}e_a\wedge \delta\left[(\psi\overline{\phi})_p+(\phi\overline{\psi})_p\right]
\end{equation}
and hence from (17) it is a CCKY $p$-form. Thus, we show that the $p$-form bilinear $(\psi\overline{\phi})_p+(\phi\overline{\psi})_p$ is a KY $p$-form or a CCKY $p$-form depending on the chosen set of $\mathcal{J}$, $\lambda$, $j$ and $p$. Note that for the chosen set of parameters, if the even degree $p$-form bilinears are KY forms then the odd degree $p$-form bilinears correspond to CCKY forms and conversely if the odd degree $p$-form bilinears are KY forms then the even degree $p$-form bilinears correspond to CCKY forms.
\end{proof}

By using Definition 8, we can refine the Proposition 1 as in the following form.
\begin{proposition}
For two Killing spinors $\psi$ and $\phi$ satisfying (11), the symmetric combination of $p$-form bilinears $(\psi\overline{\phi})_p+(\phi\overline{\psi})_p$ corresponds to a special KY $p$-form or a special CCKY $p$-form depending on the chosen involution of the inner product, on the real or imaginary character of $\lambda$ and on the parity (evenness or oddness) of $p$.
\end{proposition}
\begin{proof}
If $(\psi\overline{\phi})_p+(\phi\overline{\psi})_p$ satisfies Case 1 in (27), then it is a KY $p$-form. We can calculate the covariant derivative of $d\left[(\psi\overline{\phi})_p+(\phi\overline{\psi})_p\right]$ from (27) as
\begin{equation}
\nabla_{X_a}d\left[(\psi\overline{\phi})_p+(\phi\overline{\psi})_p\right]=2\lambda(p+1)\nabla_{X_a}\left[(\psi\overline{\phi})_{p+1}+(\phi\overline{\psi})_{p+1}\right].
\end{equation}
We know that if $(\psi\overline{\phi})_p+(\phi\overline{\psi})_p$ is a KY $p$-form, then the one higher degree bilinear $(\psi\overline{\phi})_{p+1}+(\phi\overline{\psi})_{p+1}$ satisfies Case 2 and is a CCKY $(p+1)$-form. Thus, from (29), we obtain
\begin{equation}
\nabla_{X_a}d\left[(\psi\overline{\phi})_p+(\phi\overline{\psi})_p\right]=4\lambda^2(p+1)e_a\wedge\left[(\psi\overline{\phi})_p+(\phi\overline{\psi})_p\right].
\end{equation}
This is nothing but the condition of special KY forms in (20) for $c=-4\lambda^2$. So, $(\psi\overline{\phi})_p+(\phi\overline{\psi})_p$ is a special KY $p$-form in this case.

On the other hand, if $(\psi\overline{\phi})_p+(\phi\overline{\psi})_p$ satisfies Case 2 in (29), then it is a CCKY $p$-form. We can calculate the covariant derivative of $\delta\left[(\psi\overline{\phi})_p+(\phi\overline{\psi})_p\right]$ from (29) as
\begin{equation}
\nabla_{X_a}\delta\left[(\psi\overline{\phi})_p+(\phi\overline{\psi})_p\right]=-2\lambda(n-p+1)\nabla_{X_a}\left[(\psi\overline{\phi})_{p-1}+(\phi\overline{\psi})_{p-1}\right].
\end{equation}
We know that if $(\psi\overline{\phi})_p+(\phi\overline{\psi})_p$ is a CCKY $p$-form, then the one lower degree bilinear $(\psi\overline{\phi})_{p-1}+(\phi\overline{\psi})_{p-1}$ satisfies Case 1 and is a KY $(p-1)$-form. Thus, from (27), we obtain
\begin{equation}
\nabla_{X_a}\delta\left[(\psi\overline{\phi})_p+(\phi\overline{\psi})_p\right]=-4\lambda^2(n-p+1)i_{X_a}\left[(\psi\overline{\phi})_p+(\phi\overline{\psi})_p\right].
\end{equation}
This is nothing but the condition of special CCKY forms in (21) for $c=-4\lambda^2$. So, $(\psi\overline{\phi})_p+(\phi\overline{\psi})_p$ is a special CCKY $p$-form in this case.
\end{proof}

Since all KY forms and CCKY forms are special in constant curvature manifolds, we have the following
\begin{corollary}
Killing spinors generate, via the correspondence in Proposition 2, special KY and special CCKY forms and especially all KY and CCKY forms in constant curvature manifolds. In general, non-special KY and CCKY forms cannot be generated by Killing spinors.
\end{corollary}

\section{Lie Algebras of Special KY and Special CCKY forms}

Killing vector fields and conformal Killing vector fields, which correspond to the metric duals of KY 1-forms and CKY 1-forms respectively, satisfy Lie algebra structures with respect to the Lie bracket of vector fields that is written for any two vector fields $X$ and $Y$ as
\begin{equation}
[X,Y]=\nabla_XY-\nabla_YX
\end{equation}
This bracket can be generalized to the Schouten-Nijenhuis (SN) bracket defined for the higher-degree differential forms. For any $p$-form $\alpha$ and a $q$-form $\beta$, the SN bracket is written as follows
\begin{equation}
[\alpha,\beta]_{SN}=i_{X^a}\alpha\wedge\nabla_{X_a}\beta+(-1)^{pq}i_{X^a}\beta\wedge\nabla_{X_a}\alpha
\end{equation}
which gives a $(p+q-1)$-form and reduces to the Lie bracket of vector fields (35) for $p=q=1$. It satisfies the following graded Lie bracket properties
\begin{eqnarray}
[\alpha,\beta]_{SN}&=&(-1)^{pq}[\beta,\alpha]_{SN}\\
(-1)^{p(r+1)}[\alpha,[\beta,\gamma]_{SN}]_{SN}&+&(-1)^{q(p+1)}[\beta,[\gamma,\alpha]_{SN}]_{SN}\nonumber\\
&+&(-1)^{r(q+1)}[\gamma,[\alpha,\beta]_{SN}]_{SN}=0
\end{eqnarray}
where $\gamma$ is a $r$-form. It is known that KY forms satisfy a graded Lie algebra structure with respect to the SN bracket in constant curvature manifolds \cite{Kastor Ray Traschen,Ertem1}. Since all KY forms in constant curvature manifolds are special KY forms, the graded Lie algebra property can be generalized to the special KY forms in all manifolds. So, we have
\begin{proposition}
On a manifold $M$, special KY forms satisfying (15) and (20) constitute a graded Lie algebra structure with respect to the SN bracket defined in (36).
\end{proposition}
\begin{proof}
For a KY $p$-form $\omega_1$ and a KY $q$-form $\omega_2$, the SN bracket is written as
\begin{eqnarray}
[\omega_1,\omega_2]_{SN}&=&i_{X^a}\omega_1\wedge\nabla_{X_a}\omega_2+(-1)^{pq}i_{X^a}\omega_2\wedge\nabla_{X_a}\omega_1\nonumber\\
&=&\frac{1}{q+1}i_{X^a}\omega_1\wedge i_{X_a}d\omega_2+\frac{(-1)^p}{p+1}i_{X^a}d\omega_1\wedge i_{X_a}\omega_2
\end{eqnarray}
where we have used (15). If we apply the covariant derivative $\nabla_{X_a}$ with respect to a frame basis vector $X_a$ to (39), we find
\begin{eqnarray}
\nabla_{X_a}[\omega_1,\omega_2]_{SN}&=&\frac{1}{q+1}\bigg(\nabla_{X_a}i_{X^b}\omega_1\wedge i_{X_b}d\omega_2+i_{X^b}\omega_1\wedge\nabla_{X_a}i_{X_b}d\omega_2\bigg)\nonumber\\
&&+\frac{(-1)^p}{p+1}\bigg(\nabla_{X_a}i_{X^b}d\omega_1\wedge i_{X_b}\omega_2+i_{X^b}d\omega_1\wedge\nabla_{X_a}i_{X_b}\omega_2\bigg).\nonumber\\
\end{eqnarray}
In general, we have the relation $[\nabla_X,i_Y]=i_{\nabla_XY}$ for any vector fields $X$ and $Y$. If we choose the normal coordinate frame $\{X_a\}$, then this relation transforms into $[i_{X_a},\nabla_{X_b}]=0$ since we have $\nabla_{X_a}X_b=0$ in that case. By using the normal coordinates and considering that $\omega_1$ and $\omega_2$ satisfy the equations (15) and (20), we obtain
\begin{eqnarray}
\nabla_{X_a}[\omega_1,\omega_2]_{SN}&=&\frac{1}{q+1}\bigg(\frac{1}{p+1}i_{X^b}i_{X_a}d\omega_1\wedge i_{X_b}d\omega_2\nonumber\\
&&-c(q+1)i_{X^b}\omega_1\wedge i_{X_b}\left(e_a\wedge\omega_2\right)\bigg)\nonumber\\
&&+\frac{(-1)^p}{p+1}\bigg(-c(p+1)i_{X^b}\left(e_a\wedge\omega_1\right)\wedge i_{X_b}\omega_2\nonumber\\
&&+\frac{1}{q+1}i_{X^b}d\omega_1\wedge i_{X_b}i_{X_a}d\omega_2\bigg)\nonumber\\
&=&-\frac{1}{(p+1)(q+1)}i_{X_a}\left(i_{X^b}d\omega_1\wedge i_{X_b}d\omega_2\right)-ci_{X_a}\left(\omega_1\wedge\omega_2\right)\nonumber\\
\end{eqnarray}
where we have used $i_{X^b}e_a=\delta^b_a$ and the anti-derivative property of $i_{X^b}$ in the last line. If we take the wedge product of (41) with $e^a\wedge$ from the left and use the equality $e^a\wedge\nabla_{X_a}=d$ for the torsion-free connection $\nabla_{X_a}$, we find
\begin{equation}
d[\omega_1,\omega_2]_{SN}=-\frac{p+q}{(p+1)(q+1)}i_{X^b}d\omega_1\wedge i_{X_b}d\omega_2-c(p+q)\omega_1\wedge\omega_2
\end{equation}
where we have used the equality $e^a\wedge i_{X_a}\alpha=p\alpha$ for any $p$-form $\alpha$. By taking the interior derivative $i_{X_a}$ of (42) with respect to $X_a$ and divide it with $(p+q)$, one finds
\begin{equation}
\frac{1}{p+q}i_{X_a}d[\omega_1,\omega_2]_{SN}=-\frac{1}{(p+1)(q+1)}i_{X_a}\left(i_{X^b}d\omega_1\wedge i_{X_b}d\omega_2\right)-ci_{X_a}\left(\omega_1\wedge\omega_2\right).
\end{equation}
If we compare the right hand sides of (41) and (43), we can easily see that
\begin{equation}
\nabla_{X_a}[\omega_1,\omega_2]_{SN}=\frac{1}{p+q}i_{X_a}d[\omega_1,\omega_2]_{SN}
\end{equation}
which is the KY equation (15) for the $(p+q-1)$-form $[\omega_1,\omega_2]_{SN}$. Then, we show that for a special KY $p$-form $\omega_1$ and a special KY $q$-form $\omega_2$, their SN bracket $[\omega_1,\omega_2]_{SN}$ is a KY $(p+q-1)$-form. However, we still have to show that it is a special KY $(p+q-1)$-form. To show this, we take the covariant derivative of (42) which is equal to
\begin{eqnarray}
\nabla_{X_a}d[\omega_1,\omega_2]_{SN}&=&-\frac{p+q}{(p+1)(q+1)}\bigg(i_{X^b}\nabla_{X_a}d\omega_1\wedge i_{X_b}d\omega_2\nonumber\\
&&+i_{X^b}d\omega_1\wedge i_{X_b}\nabla_{X_a}d\omega_2\bigg)\nonumber\\
&&-c(p+q)\bigg(\nabla_{X_a}\omega_1\wedge\omega_2+\omega_1\wedge\nabla_{X_a}\omega_2\bigg)
\end{eqnarray}
where we have used $[i_{X^b},\nabla_{X_a}]=0$ in normal coordinates. Since $\omega_1$ and $\omega_2$ satisfy (15) and (20), we find
\begin{eqnarray}
\nabla_{X_a}d[\omega_1,\omega_2]_{SN}&=&-\frac{p+q}{(p+1)(q+1)}\bigg(-c(p+1)i_{X^b}(e_a\wedge\omega_1)\wedge i_{X_b}d\omega_2\nonumber\\
&&-c(q+1)i_{X^b}d\omega_1\wedge i_{X_b}(e_a\wedge\omega_2)\bigg)\nonumber\\
&&-c(p+q)\bigg(\frac{1}{p+1}i_{X_a}d\omega_1\wedge\omega_2+\frac{1}{q+1}\omega_1\wedge i_{X_a}d\omega_2\bigg)\nonumber\\
&=&-c(p+q)\bigg(\frac{1}{q+1}e_a\wedge i_{X^b}\omega_1\wedge i_{X_b}d\omega_2\nonumber\\
&&+\frac{(-1)^p}{p+1}e_a\wedge i_{X^b}d\omega_1\wedge i_{X_b}\omega_2\bigg)
\end{eqnarray}
where we have used $i_{X^b}e_a=\delta^b_a$ and the anti-derivative property of $i_{X^b}$ in the last line. By comparing with (39), we obtain
\begin{equation}
\nabla_{X_a}d[\omega_1,\omega_2]_{SN}=-c(p+q)e_a\wedge [\omega_1,\omega_2]_{SN}
\end{equation}
which is the equation (20) for KY $(p+q-1)$-form $[\omega_1,\omega_2]_{SN}$. This shows that $[\omega_1,\omega_2]_{SN}$ is a special KY form and hence special KY forms satisfy a graded Lie algebra structure under SN bracket.
\end{proof}

For the more general case of CKY forms, we can also define a graded Lie bracket which generalizes the SN bracket. For a CKY $p$-form $\omega_1$ and a CKY $q$-form $\omega_2$, we define the CKY bracket as follows
\begin{eqnarray}
[\omega_1,\omega_2]_{CKY}&=&\frac{1}{q+1}i_{X_a}\omega_1\wedge i_{X^a}d\omega_2+\frac{(-1)^p}{p+1}i_{X_a}d\omega_1\wedge i_{X^a}\omega_2\nonumber\\
&&+\frac{(-1)^p}{n-q+1}\omega_1\wedge\delta\omega_2+\frac{1}{n-p+1}\delta\omega_1\wedge\omega_2.
\end{eqnarray}
It is shown in \cite{Ertem3} that in constant curvature manifolds $[\omega_1,\omega_2]_{CKY}$ is a CKY $(p+q-1)$-form and CKY bracket has the graded Lie bracket properties. So, in constant curvature manifolds, CKY forms satisfy a graded Lie algebra structure with respect to the CKY bracket and we have the following relation
\begin{equation}
\nabla_{X_a}[\omega_1,\omega_2]_{CKY}=\frac{1}{p+q}i_{X_a}d[\omega_1,\omega_2]_{CKY}-\frac{1}{n-p-q+2}e_a\wedge\delta[\omega_1,\omega_2]_{CKY}.
\end{equation}
This is also true in Einstein manifolds for normal CKY forms which have the following special integrabilitiy conditions for a normal CKY $p$-form $\omega$ in $n$ dimensions
\begin{eqnarray}
\nabla_{X_a}d\omega&=&\frac{p+1}{p(n-p+1)}e_a\wedge d\delta\omega+2(p+1)K_a\wedge\omega\\
\nabla_{X_a}\delta\omega&=&-\frac{n-p+1}{(p+1)(n-p)}i_{X_a}\delta d\omega-2(n-p+1)i_{X^b}K_a\wedge i_{X_b}\omega
\end{eqnarray}
where the Schouten rho 1-form is defined by
\begin{equation}
K_a=\frac{1}{n-2}\left(\frac{\mathcal{R}}{2(n-1)}e_a-P_a\right)
\end{equation}
for Ricci 1-forms $P_a$ and the scalar curvature ${\mathcal{R}}$. It can easily be seen that for two KY forms $\omega_1$ and $\omega_2$ with $\delta\omega_1=0$ and $\delta\omega_2=0$, the CKY bracket reduces to the SN bracket given in (39) and for $p=q=1$ it reduces to the Lie bracket of vector fields in (35). However, for general CKY forms, it differs from the generalization of the SN bracket to CKY forms and hence is a new bracket.

For more general manifolds, there is another graded Lie algebra structure with respect to the CKY bracket that contains special KY and special CCKY forms as elements which are the special forms generated by Killing spinors. Note that, it can be seen from Propositions 1 and 2 that if the odd degree $p$-form Dirac currents of a Killing spinor are special KY forms, then its even degree $p$-form Dirac currents correspond to special CCKY forms and conversely, if the even degree $p$-form Dirac currents of a Killing spinor are special KY forms, then its odd degree $p$-form Dirac currents correspond to special CCKY forms.

\begin{theorem}
On a manifold $M$, odd degree special KY forms satisfying (15) and (20) and even degree special CCKY forms satisfying (17) and (21) constitute a Lie algebra structure with respect to the CKY bracket $[\,,\,]_{CKY}$ given in (48) in all dimensions.
\end{theorem}
\begin{proof}
We consider the cases of two special odd KY forms, two special even CCKY forms and one special odd KY form with one special even CCKY form separately. For a special odd KY $p$-form $\omega_1$ and a special odd KY $q$-form $\omega_2$, we have $\delta\omega_1=0$ and $\delta\omega_2=0$ and the CKY bracket in (48) reduces to the SN bracket written in (39) for special KY forms. From Proposition 3, $[\omega_1,\omega_2]_{SN}$ is also a special odd KY $(p+q-1)$-form since $p+q-1$ is an odd number for $p$ and $q$ are odd.

If we take a special odd KY $p$-form $\omega_1$ and a special even CCKY $q$-form $\omega_2$, then we have $\delta\omega_1=0$ and $d\omega_2=0$. So, the CKY bracket reduces to
\begin{equation}
[\omega_1,\omega_2]_{CKY}=-\frac{1}{p+1}i_{X^a}d\omega_1\wedge i_{X_a}\omega_2-\frac{1}{n-q+1}\omega_1\wedge\delta\omega_2.
\end{equation}
The covariant derivative of (53) gives
\begin{eqnarray}
\nabla_{X_a}[\omega_1,\omega_2]_{CKY}&=&-\frac{1}{p+1}i_{X^b}\nabla_{X_a}d\omega_1\wedge i_{X_b}\omega_2\nonumber\\
&&-\frac{1}{p+1}i_{X^b}d\omega_1\wedge i_{X_b}\nabla_{X_a}\omega_2\nonumber\\
&&-\frac{1}{n-q+1}\nabla_{X_a}\omega_1\wedge\delta\omega_2\nonumber\\
&&-\frac{1}{n-q+1}\omega_1\wedge\nabla_{X_a}\delta\omega_2.
\end{eqnarray}
If we use (15) and (20) for $\omega_1$ and (17) and (21) for $\omega_2$, we find
\begin{eqnarray}
\nabla_{X_a}[\omega_1,\omega_2]_{CKY}&=&ci_{X^b}(e_a\wedge\omega_1)\wedge i_{X^b}\omega_2\nonumber\\
&&+\frac{1}{(p+1)(n-q+1)}i_{X^b}d\omega_1\wedge i_{X_b}(e_a\wedge\delta\omega_2)\nonumber\\
&&-\frac{1}{(p+1)(n-q+1)}i_{X_a}d\omega_1\wedge\delta\omega_2-c\omega_1\wedge i_{X_a}\omega_2\nonumber\\
&=&-ce_a\wedge i_{X^b}\omega_1\wedge i_{X_b}\omega_2\nonumber\\
&&+\frac{1}{(p+1)(n-q+1)}e_a\wedge i_{X^b}d\omega_1\wedge i_{X_b}\delta\omega_2.
\end{eqnarray}
By taking the wedge product with $e^a\wedge$ from the left, we have
\begin{equation}
d[\omega_1,\omega_2]_{CKY}=e^a\wedge\nabla_{X_a}[\omega_1,\omega_2]_{CKY}=0.
\end{equation}
Moreover, if we calculate the interior derivative of (54) with respect to $X^a$, we find
\begin{eqnarray}
\delta[\omega_1,\omega_2]_{CKY}&=&-i_{X^a}\nabla_{X_a}[\omega_1,\omega_2]_{CKY}\nonumber\\
&=&c(n-(p+q-2))(i_{X^b}\omega_1\wedge i_{X_b}\omega_2)\nonumber\\
&&-\frac{n-(p+q-2)}{(p+1)(n-q+1)}i_{X^b}d\omega_1\wedge i_{X_b}\delta\omega_2.
\end{eqnarray}
Hence, by comparing (55) and (57), we obtain
\begin{equation}
\nabla_{X_a}[\omega_1,\omega_2]_{CKY}=-\frac{1}{n-(p+q-2)}e_a\wedge\delta[\omega_1,\omega_2]_{CKY}.
\end{equation}
This means that $[\omega_1,\omega_2]_{CKY}$ is an even CCKY $(p+q-1)$-form since for $p$ odd and $q$ even, $p+q-1$ is an even number. To prove the specialty, we need to apply the covariant derivative $\nabla_{X_a}$ to (57) and this gives us
\begin{eqnarray}
\nabla_{X_a}\delta[\omega_1,\omega_2]_{CKY}&=&c(n-(p+q-2))\big(i_{X^b}\nabla_{X_a}\omega_1\wedge i_{X_b}\omega_2\nonumber\\
&&+i_{X^b}\omega_1\wedge i_{X_b}\nabla_{X_a}\omega_2\big)\nonumber\\
&&-\frac{n-(p+q-2)}{(p+1)(n-q+1)}\big(i_{X^b}\nabla_{X_a}d\omega_1\nonumber\\
&&+i_{X^b}d\omega_1\wedge i_{X_b}\nabla_{X_a}\delta\omega_2\big)\nonumber\\
&=&\frac{c(n-(p+q-2))}{p+1}i_{X^b}i_{X_a}d\omega_1\wedge i_{X_b}\omega_2\nonumber\\
&&-\frac{c(n-(p+q-2))}{n-q+1}i_{X_a}\omega_1\wedge\delta\omega_2\nonumber\\
&&+\frac{c(n-(p+q-2))}{n-q+1}\omega_1\wedge i_{X_a}\delta\omega_2\nonumber\\
&&-\frac{c(n-(p+q-2))}{p+1}i_{X^b}d\omega_1\wedge i_{X_b}i_{X_a}\omega_2
\end{eqnarray}
where we have used (15) and (20) for $\omega_1$ and (17) and (21) for $\omega_2$. By comparing (59) with (53), one can easily see that
\begin{equation}
\nabla_{X_a}\delta[\omega_1,\omega_2]_{CKY}=c(n-(p+q-2))i_{X_a}[\omega_1,\omega_2]_{CKY}.
\end{equation}
This means that the CCKY $(p+q-1)$-form $[\omega_1,\omega_2]_{CKY}$ is a special even CCKY form from (21).

On the other hand, if we take $\omega_1$ as a special even CCKY $p$-form and $\omega_2$ as a special even CCKY $q$-form, the CCKY bracket is written as
\begin{equation}
[\omega_1,\omega_2]_{CKY}=\frac{1}{n-q+1}\omega_1\wedge\delta\omega_2+\frac{1}{n-p+1}\delta\omega_1\wedge\omega_2.
\end{equation}
The covariant derivative of (61) gives
\begin{eqnarray}
\nabla_{X_a}[\omega_1,\omega_2]_{CKY}&=&\frac{1}{n-q+1}\left(\nabla_{X_a}\omega_1\wedge\delta\omega_2+\omega_1\wedge\nabla_{X_a}\delta\omega_2\right)\nonumber\\
&&+\frac{1}{n-p+1}\left(\nabla_{X_a}\delta\omega_1\wedge\omega_2+\delta\omega_1\wedge\nabla_{X_a}\omega_2\right)\nonumber\\
&=&-\frac{1}{(n-p+1)(n-q+1)}e_a\wedge\delta\omega_1\wedge\delta\omega_2\nonumber\\
&&+c\omega_1\wedge i_{X_a}\omega_2+ci_{X_a}\omega_1\wedge\omega_2\nonumber\\
&&+\frac{1}{(n-p+1)(n-q+1)}e_a\wedge\delta\omega_1\wedge\delta\omega_2\nonumber\\
&=&ci_{X_a}\left(\omega_1\wedge\omega_2\right)
\end{eqnarray}
where we have used (17) and (21). By taking the interior derivative with respect to $X^a$, we find
\begin{equation}
\delta[\omega_1,\omega_2]_{CKY}=-i_{X^a}\nabla_{X_a}[\omega_1,\omega_2]_{CKY}=0
\end{equation}
and the wedge product of (62) with $e^a\wedge$ from the left gives
\begin{equation}
d[\omega_1,\omega_2]_{CKY}=e^a\wedge\nabla_{X_a}[\omega_1,\omega_2]_{CKY}=c(p+q)\omega_1\wedge\omega_2.
\end{equation}
So, by comparing (62) and (64), we obtain
\begin{equation}
\nabla_{X_a}[\omega_1,\omega_2]_{CKY}=\frac{1}{p+q}i_{X_a}d[\omega_1,\omega_2]_{CKY}.
\end{equation}
This shows that $[\omega_1,\omega_2]_{CKY}$ is an odd KY $(p+q-1)$-form since $p+q-1$ is an odd number for $p$ and $q$ are even. To prove the specialty, we take the covariant derivative of (64) which gives
\begin{eqnarray}
\nabla_{X_a}d[\omega_1,\omega_2]_{CKY}&=&c(p+q)\left(\nabla_{X_a}\omega_1\wedge\omega_2+\omega_1\wedge\nabla_{X_a}\omega_2\right)\nonumber\\
&=&c(p+q)\bigg(-\frac{1}{n-p+1}e_a\wedge\delta\omega_1\wedge\omega_2\nonumber\\
&&-\frac{1}{n-q+1}\omega_1\wedge e_a\wedge\delta\omega_2\bigg)
\end{eqnarray}
where we have used (17). By comparing with (61), one can see that
\begin{equation}
\nabla_{X_a}d[\omega_1,\omega_2]_{CKY}=-c(p+q)e_a\wedge[\omega_1,\omega_2]_{CKY}.
\end{equation}
So, the KY $(p+q-1)$-form $[\omega_1,\omega_2]_{CKY}$ is a special odd KY form from (20).

In all cases, the CKY bracket has the property $[\omega_1,\omega_2]_{CKY}=-[\omega_2,\omega_1]_{CKY}$ and satisfies the Jacobi identity with no restriction on the dimension. Hence, we prove that the special odd KY forms and special even CCKY forms satisfy a Lie algebra structure with respect to the CKY bracket such that two special odd KY forms gives again a special odd KY form, one special odd KY form and one special even CCKY form gives another special even CCKY form, two special even CCKY forms give a special odd KY form which is shown diagrammatically as
\begin{eqnarray}
 \textrm{special odd KY}\times\textrm{special odd KY}&\longrightarrow&\textrm{special odd KY}\nonumber\\
 \textrm{special odd KY}\times\textrm{special even CCKY}&\longrightarrow&\textrm{special even CCKY}\nonumber\\
 \textrm{special even CCKY}\times\textrm{special even CCKY}&\longrightarrow&\textrm{special odd KY}\nonumber
\end{eqnarray}
\end{proof}

On the other hand, if the Killing spinors on a manifold $M$ generate even degree special KY forms and odd degree special CCKY forms, then the Lie algebra structure is dependent on the dimension as stated in the following

\begin{theorem}
On a manifold $M$, even degree special KY forms satisfying (15) and (20) and odd degree special CCKY forms satisfying (17) and (21) constitute a Lie algebra structure with respect to the Hodge star of the CKY bracket $*[\,,\,]_{CKY}$ given in (48) in even dimensions.
\end{theorem}
\begin{proof}
We use the property that KY forms and CCKY forms are Hodge duals of each other. We consider the cases of two special even KY forms, two special odd CCKY forms and one special even KY form with one special odd CCKY form separately in even dimensions. For a special even KY $p$-form $\omega_1$ and a special even KY $q$-form $\omega_2$, we have $\delta\omega_1=0$ and $\delta\omega_2=0$ and the CKY bracket in (48) reduces to the SN bracket for special KY forms. $[\omega_1,\omega_2]_{SN}$ corresponds to a special odd KY $(p+q-1)$-form since $p+q-1$ is an odd number for $p$ and $q$ are even. However, the Hodge dual of a special odd KY form is a special odd CCKY form in even $n$ dimensions. So, $*[\omega_1,\omega_2]_{CKY}=*[\omega_1,\omega_2]_{SN}$ is a special odd CCKY $(n-(p+q-1))$-form.

If we take a special even KY $p$-form $\omega_1$ and a special odd CCKY $q$-form $\omega_2$, then we have $\delta\omega_1=0$ and $d\omega_2=0$. So, the CKY bracket reduces to
\begin{equation}
[\omega_1,\omega_2]_{CKY}=\frac{1}{p+1}i_{X^a}d\omega_1\wedge i_{X_a}\omega_2+\frac{1}{n-q+1}\omega_1\wedge\delta\omega_2.
\end{equation}
From the same considerations in (54) and (55), the covariant derivative of (68) gives
\begin{equation}
\nabla_{X_a}[\omega_1,\omega_2]_{CKY}=ce_a\wedge i_{X^b}\omega_1\wedge i_{X_b}\omega_2+\frac{1}{(p+1)(n-q+1)}e_a\wedge i_{X^b}d\omega_1\wedge i_{X_b}\delta\omega_2.
\end{equation}
By taking the wedge product with $e^a\wedge$ from the left, we find
\begin{equation}
d[\omega_1,\omega_2]_{CKY}=0
\end{equation}
and the interior derivative of (69) gives
\begin{eqnarray}
\delta[\omega_1,\omega_2]_{CKY}&=&-c(n-(p+q-2))i_{X^b}\omega_1\wedge i_{X_b}\omega_2\nonumber\\
&&-\frac{n-(p+q-2)}{(p+1)(n-q+1)}i_{X^b}d\omega_1\wedge i_{X_b}\delta\omega_2.
\end{eqnarray}
Hence, by comparing (55) and (57), we obtain
\begin{equation}
\nabla_{X_a}[\omega_1,\omega_2]_{CKY}=-\frac{1}{n-(p+q-2)}e_a\wedge\delta[\omega_1,\omega_2]_{CKY}.
\end{equation}
This means that $[\omega_1,\omega_2]_{CKY}$ is an even CCKY $(p+q-1)$-form since for $p$ even and $q$ odd, $p+q-1$ is an even number. So, from the Hodge duality property between KY and CCKY forms, $*[\omega_1,\omega_2]_{CKY}$ is an even KY $(n-(p+q-1))$-form in even $n$ dimensions. To prove the specialty, we apply the covariant derivative $\nabla_{X_a}$ to (71) and this gives
\begin{eqnarray}
\nabla_{X_a}\delta[\omega_1,\omega_2]_{CKY}&=&-\frac{c(n-(p+q-2))}{p+1}i_{X^b}i_{X_a}d\omega_1\wedge i_{X_b}\omega_2\nonumber\\
&&+\frac{c(n-(p+q-2))}{n-q+1}i_{X_a}\omega_1\wedge\delta\omega_2\nonumber\\
&&+\frac{c(n-(p+q-2))}{n-q+1}\omega_1\wedge i_{X_a}\delta\omega_2\nonumber\\
&&-\frac{c(n-(p+q-2))}{p+1}i_{X^b}d\omega_1\wedge i_{X_b}i_{X_a}\omega_2
\end{eqnarray}
by using the same procedures as in (59). The comparison of (73) with (68) shows that we have
\begin{equation}
\nabla_{X_a}\delta[\omega_1,\omega_2]_{CKY}=c(n-(p+q-2))i_{X_a}[\omega_1,\omega_2]_{CKY}
\end{equation}
and this means that $[\omega_1,\omega_2]_{CKY}$ is a special even CCKY $(p+q-1)$-form and $*[\omega_1,\omega_2]_{CKY}$ is a special even KY $(n-(p+q-1))$-form in even dimensions.

On the other hand, if we take $\omega_1$ as a special odd CCKY $p$-form and $\omega_2$ as a special odd CCKY $q$-form, the CKY bracket corresponds to
\begin{equation}
[\omega_1,\omega_2]_{CKY}=-\frac{1}{n-q+1}\omega_1\wedge\delta\omega_2+\frac{1}{n-p+1}\delta\omega_1\wedge\omega_2
\end{equation}
and the covariant derivative of it gives
\begin{equation}
\nabla_{X_a}[\omega_1,\omega_2]_{CKY}=ci_{X_a}(\omega_1\wedge\omega_2).
\end{equation}
By using (76), the exterior and co-derivatives can be found as
\begin{eqnarray}
d[\omega_1,\omega_2]_{CKY}&=&c(p+q)\omega_1\wedge\omega_2\\
\delta[\omega_1,\omega_2]_{CKY}&=&0.
\end{eqnarray}
From the comparison of (76) and (77), one can see that
\begin{equation}
\nabla_{X_a}[\omega_1,\omega_2]_{CKY}=\frac{1}{p+q}i_{X_a}d[\omega_1,\omega_2]_{CKY}
\end{equation}
and hence $[\omega_1,\omega_2]_{CKY}$ is an odd KY $(p+q-1)$-form for $p$ and $q$ are odd. This means that $*[\omega_1,\omega_2]_{CKY}$ is an odd CCKY $(n-(p+q-1))$-form in even $n$ dimensions. The covariant derivative of (77) gives
\begin{equation}
\nabla_{X_a}d[\omega_1,\omega_2]_{CKY}=c(p+q)e_a\wedge\left(-\frac{1}{n-p+1}\delta\omega_1\wedge\omega_2+\frac{1}{n-q+1}\omega_1\wedge\delta\omega_2\right)
\end{equation}
and we have
\begin{equation}
\nabla_{X_a}d[\omega_1,\omega_2]_{CKY}=-c(p+q)e_a\wedge[\omega_1,\omega_2]_{CKY}.
\end{equation}
So, $[\omega_1,\omega_2]_{CKY}$ is a special odd KY $(p+q-1)$-form and $*[\omega_1,\omega_2]_{CKY}$ is a special odd CCKY $(n-(p+q-1))$-form in even $n$ dimensions.

In all cases, Hodge star of the CKY bracket has the property $*[\omega_1,\omega_2]_{CKY}=-*[\omega_2,\omega_1]_{CKY}$ and satisfies the Jacobi identity. Hence, we prove that the special even KY forms and special odd CCKY forms satisfy a Lie algebra structure with respect to the bracket $*[\,,\,]_{CKY}$ in even dimensions such that two special even KY forms gives a special odd CCKY form, one special even KY form and one special odd CCKY form gives another special even KY form, two special odd CCKY forms give again a special odd CCKY form which is shown diagrammatically as
\begin{eqnarray}
 \textrm{special even KY}\times\textrm{special even KY}&\longrightarrow&\textrm{special odd CCKY}\nonumber\\
 \textrm{special even KY}\times\textrm{special odd CCKY}&\longrightarrow&\textrm{special even KY}\nonumber\\
 \textrm{special odd CCKY}\times\textrm{special odd CCKY}&\longrightarrow&\textrm{special odd CCKY}\nonumber
\end{eqnarray}
\end{proof}

In odd dimensions, special even KY forms and special odd CCKY forms do not satisfy a Lie algebra with respect to the bracket $*[\,,\,]_{CKY}$. In that case, two special even KY forms give a special even CCKY form, one special even KY form and one special odd CCKY form give a special odd KY form, two special odd CCKY forms give a special even CCKY form. So, special even KY forms and special odd CCKY forms does not close into an algebra with respect to $*[\,,\,]_{CKY}$. This is also true for the bracket $[\,,\,]_{CKY}$. Unlike the case of special odd KY and special even CCKY forms, the algebra structure of special even KY and special odd CCKY forms is only relevant in even dimensions.

\section{Symmetry Operators}

Killing spinors and KY forms are also related to each other in another way. For a Killing vector field $K$, the Lie derivative $\pounds_K$ acting on a spinor $\psi$ is defined as follows \cite{Kosmann}
\begin{equation}
\pounds_K\psi=\nabla_K\psi+\frac{1}{4}d\widetilde{K}.\psi
\end{equation}
where $\widetilde{K}$ is the 1-form which corresponds to the metric dual of $K$. It is known that $\pounds_K$ is a symmetry operator of the Killing spinor equation for all Killing vector fields $K$ \cite{Benn Tucker}. This means that if $\psi$ is a Killing spinor, then $\pounds_K\psi$ is also a Killing spinor, namely we have
\begin{equation}
\nabla_X\pounds_K\psi=\lambda\widetilde{X}.\pounds_K\psi
\end{equation}
for all vector fields $X$ and all Killing vector fields $K$. Moreover, if we generalize the operator in (82) to a KY $p$-form $\omega$ as
\begin{equation}
L_{\omega}=(i_{X^a}\omega).\nabla_{X_a}+\frac{p}{2(p+1)}d\omega
\end{equation}
and apply it to a Killing spinor $\psi$, we find
\begin{equation}
L_{\omega}\psi=(-1)^{p-1}\lambda p\omega.\psi+\frac{p}{2(p+1)}d\omega.\psi.
\end{equation}
If the degree $p$ of the KY $p$-form $\omega$ is odd, then we have
\begin{equation}
L_{\omega}\psi=\lambda p\omega.\psi+\frac{p}{2(p+1)}d\omega.\psi
\end{equation}
and it is shown in \cite{Ertem1} that this operator with respect to an odd KY $p$-form $\omega$ is also a symmetry operator for the Killing spinor equation in constant curvature manifolds. So, in that case we have
\begin{equation}
\nabla_XL_{\omega}\psi=\lambda\widetilde{X}.L_{\omega}\psi.
\end{equation}
The operator defined in (84) is a special case of the symmetry operator of twistor spinors in constant curvature manifolds written in terms of CKY $p$-forms $\omega$ as
\begin{equation}
L_{\omega}=(i_{X^a}\omega).\nabla_{X_a}+\frac{p}{2(p+1)}d\omega+\frac{p}{2(n-p+1)}\delta\omega
\end{equation}
which is proved in \cite{Ertem5, Ertem2}.

We show that this symmetry operator construction can be generalized to all manifolds admitting special KY and special CCKY forms.

\begin{theorem}
For every special odd KY $p$-form $\omega$ and every special even CCKY $p$-form $\omega$ generated by a Killing spinor with Killing number $\lambda$, the following operator
\begin{equation}
L_{\omega}=(i_{X^a}\omega).\nabla_{X_a}+\frac{p}{2(p+1)}d\omega+\frac{p}{2(n-p+1)}\delta\omega
\end{equation}
is a symmetry operator for the Killing spinor equation (11).
\end{theorem}
\begin{proof}
To prove the theorem, we have to show that for $\omega$ is a special odd KY $p$-form or a special even CCKY $p$-form and $\psi$ is a Killing spinor, the operator $L_{\omega}$ in (89) satisfies (87).

If $\omega$ is a special odd KY $p$-form, then we have $\delta\omega=0$ and the action of  (89) on a Killing spinor $\psi$ is written as
\begin{eqnarray}
L_{\omega}\psi&=&(i_{X^a}\omega).\nabla_{X_a}\psi+\frac{p}{2(p+1)}d\omega.\psi\nonumber\\
&=&\lambda(i_{X^a}\omega).e_a.\psi+\frac{p}{2(p+1)}d\omega.\psi\nonumber\\
&=&\lambda p\omega.\psi+\frac{p}{2(p+1)}d\omega.\psi
\end{eqnarray}
where we have used $(i_{X^a}\omega).e_a=e_a\wedge i_{X^a}\omega=p\omega$ since $p$ is odd and $i_{X_a}i_{X^a}=0$. By taking the covariant derivative $\nabla_{X_a}$ with respect to a basis $X_a$ of vector fields, we can calculate the left hand side of (87) as
\begin{eqnarray}
\nabla_{X_a}L_{\omega}\psi&=&\lambda p\left(\nabla_{X_a}\omega.\psi+\omega.\nabla_{X_a}\psi\right)+\frac{p}{2(p+1)}\left(\nabla_{X_a}d\omega.\psi+d\omega.\nabla_{X_a}\psi\right)\nonumber\\
&=&\lambda\frac{p}{p+1}i_{X_a}d\omega.\psi+\lambda^2p\omega.e_a.\psi-\frac{cp}{2}(e_a\wedge\omega).\psi\nonumber\\
&&+\lambda\frac{p}{2(p+1)}d\omega.e_a.\psi
\end{eqnarray}
where we have used (11), (15) and (20). Since the odd KY $p$-form $\omega$ is generated by a Killing spinor with the Killing number $\lambda$, we have $c=-4\lambda^2$ from (32). By using (23), we can write $\omega.e_a=-e_a\wedge\omega+i_{X_a}\omega$ and $d\omega.e_a=e_a\wedge d\omega-i_{X_a}d\omega$ since $\omega$ is an odd form. Then, we have
\begin{eqnarray}
\nabla_{X_a}L_{\omega}\psi&=&\lambda\frac{p}{p+1}i_{X_a}d\omega.\psi+\lambda^2p(-e_a\wedge\omega+i_{X_a}\omega).\psi\nonumber\\
&&+2\lambda^2p(e_a\wedge\omega).\psi+\lambda\frac{p}{2(p+1)}(e_a\wedge d\omega-i_{X_a}d\omega).\psi\nonumber\\
&=&\lambda^2p(e_a\wedge\omega+i_{X_a}\omega).\psi\nonumber\\
&&+\lambda\frac{p}{2(p+1)}(e_a\wedge d\omega+i_{X_a}d\omega).\psi
\end{eqnarray}
From (22) and the operator $L_{\omega}$ in (90), we obtain
\begin{eqnarray}
\nabla_{X_a}L_{\omega}\psi&=&\lambda^2pe_a.\omega.\psi+\lambda\frac{p}{2(p+1)}e_a.d\omega.\psi\nonumber\\
&=&\lambda e_a.L_{\omega}\psi
\end{eqnarray}
which is the equation (87) for the frame basis $X_a$ and co-frame basis $e^a$. So, $L_{\omega}$ defined in (89) with respect to a special odd KY form $\omega$ is a symmetry operator of the Killing spinor equation.

On the other hand, if $\omega$ is a special even CCKY $p$-form, then we have $d\omega=0$ and the action of  (89) on a Killing spinor $\psi$ is written as
\begin{eqnarray}
L_{\omega}\psi&=&(i_{X^a}\omega).\nabla_{X_a}\psi+\frac{p}{2(n-p+1)}\delta\omega.\psi\nonumber\\
&=&\lambda(i_{X^a}\omega).e_a.\psi+\frac{p}{2(n-p+1)}\delta\omega.\psi\nonumber\\
&=&-\lambda p\omega.\psi+\frac{p}{2(n-p+1)}\delta\omega.\psi
\end{eqnarray}
where we have used $(i_{X^a}\omega).e_a=-e_a\wedge i_{X^a}\omega=-p\omega$ since $p$ is even and $i_{X_a}i_{X^a}=0$. By taking the covariant derivative $\nabla_{X_a}$ with respect to a basis $X_a$ of vector fields, the left hand side of (87) gives
\begin{eqnarray}
\nabla_{X_a}L_{\omega}\psi&=&-\lambda p\left(\nabla_{X_a}\omega.\psi+\omega.\nabla_{X_a}\psi\right)\nonumber\\
&&+\frac{p}{2(n-p+1)}\left(\nabla_{X_a}\delta\omega.\psi+\delta\omega.\nabla_{X_a}\psi\right)\nonumber\\
&=&\lambda\frac{p}{n-p+1}(e_a\wedge\delta\omega).\psi-\lambda^2p\omega.e_a.\psi\nonumber\\
&&+\frac{cp}{2}i_{X_a}\omega.\psi+\lambda\frac{p}{2(n-p+1)}\delta\omega.e_a.\psi
\end{eqnarray}
where we have used (11), (17) and (21). Since the even CCKY $p$-form $\omega$ is generated by a Killing spinor with the Killing number $\lambda$, we have $c=-4\lambda^2$ from (34). By using (23), we can write $\omega.e_a=e_a\wedge\omega-i_{X_a}\omega$ and $\delta\omega.e_a=-e_a\wedge\delta\omega+i_{X_a}\delta\omega$ since $\omega$ is an even form. Then, we have
\begin{eqnarray}
\nabla_{X_a}L_{\omega}\psi&=&\lambda\frac{p}{n-p+1}(e_a\wedge\delta\omega).\psi-\lambda^2p(e_a\wedge\omega-i_{X_a}\omega).\psi\nonumber\\
&&-2\lambda^2pi_{X_a}\omega.\psi+\lambda\frac{p}{2(n-p+1)}(-e_a\wedge\delta\omega+i_{X_a}\delta\omega).\psi\\
&=&-\lambda^2p(e_a\wedge\omega+i_{X_a}\omega).\psi+\lambda\frac{p}{2(n-p+1)}(e_a\wedge\delta\omega+i_{X_a}\delta\omega).\psi\nonumber
\end{eqnarray}
From (22) and the operator $L_{\omega}$ in (94), we obtain
\begin{eqnarray}
\nabla_{X_a}L_{\omega}\psi&=&-\lambda^2pe_a.\omega.\psi+\lambda\frac{p}{2(n-p+1)}e_a.\delta\omega.\psi\nonumber\\
&=&\lambda e_a.L_{\omega}\psi
\end{eqnarray}
which is the equation (87) for the frame basis $X_a$ and co-frame basis $e^a$. So, $L_{\omega}$ defined in (89) with respect to a special even CCKY form $\omega$ is a symmetry operator of the Killing spinor equation.
\end{proof}

For the case of special even KY forms and special odd CCKY forms, we have
\begin{theorem}
For every special even KY $p$-form $\omega$ and every special odd CCKY $p$-form $\omega$ generated by a Killing spinor with Killing number $\lambda$, the following operator
\begin{equation}
K_{\omega}=(i_{X^a}\omega).z.\nabla_{X_a}+\frac{p}{2(p+1)}d\omega.z+\frac{p}{2(n-p+1)}\delta\omega.z
\end{equation}
is a symmetry operator for the Killing spinor equation (11) in even $n$ dimensions. Here $z$ is the volume form.
\end{theorem}
\begin{proof}
To prove the theorem, we need to show that for $\omega$ is a special even KY $p$-form or a special odd CCKY $p$-form and $\psi$ is a Killing spinor, the operator $K_{\omega}$ in (98) satisfies
\begin{equation}
\nabla_{X_a}K_{\omega}\psi=\lambda e_a.K_{\omega}\psi.
\end{equation}

If $\omega$ is a special even KY $p$-form, then we have $\delta\omega=0$ and (98) turns into
\begin{eqnarray}
K_{\omega}\psi&=&(i_{X^a}\omega).z.\nabla_{X_a}\psi+\frac{p}{2(p+1)}d\omega.z.\psi\nonumber\\
&=&\lambda(i_{X^a}\omega).z.e_a.\psi+\frac{p}{2(p+1)}d\omega.z.\psi
\end{eqnarray}
where we have used the Killing spinor equation (11). In even dimensions, the volume form $z$ anticommutes with all of the generators of the Clifford algebra and we have $z.e_a=-e_a.z$. Moreover, since $\omega$ is an even form, $i_{X^a}\omega$ is an odd form and we have $i_{X^a}\omega.e_a=-e_a\wedge i_{X^a}\omega=-p\omega$. So, we obtain
\begin{eqnarray}
K_{\omega}\psi&=&-\lambda(i_{X^a}\omega).e_a.z.\psi+\frac{p}{2(p+1)}d\omega.z.\psi\nonumber\\
&=&\lambda p\omega.z.\psi+\frac{p}{2(p+1)}d\omega.z.\psi.
\end{eqnarray}
To find the left hand side of (99), we can take the covariant derivative of (101) with respect to $X_a$ and this gives
\begin{eqnarray}
\nabla_{X_a}K_{\omega}\psi&=&\lambda p\left(\nabla_{X_a}\omega.z.\psi+\omega.z.\nabla_{X_a}\psi\right)\nonumber\\
&&+\frac{p}{2(p+1)}\left(\nabla_{X_a}d\omega.z.\psi+d\omega.z.\nabla_{X_a}\psi\right)\nonumber\\
&=&\lambda\frac{p}{p+1}i_{X_a}d\omega.z.\psi+\lambda^2p\omega.z.e_a.\psi-\frac{cp}{2}(e_a\wedge\omega).z.\psi\nonumber\\
&&+\lambda\frac{p}{2(p+1)}d\omega.z.e_a.\psi
\end{eqnarray}
where we have used the fact that $\nabla$ is metric compatible, so $\nabla_{X_a}z=0$ and the equations (11), (15) and (20). Since the even KY $p$-form $\omega$ is generated by a Killing spinor with the Killing number $\lambda$, we have $c=-4\lambda^2$ from (32). By using the relation $z.e_a=-e_a.z$ and the identities $\omega.e_a=e_a\wedge\omega-i_{X_a}\omega$ and $d\omega.e_a=-e_a\wedge d\omega+i_{X_a}d\omega$ resulting from (23) since $\omega$ is an even form, we obtain
\begin{eqnarray}
\nabla_{X_a}K_{\omega}\psi&=&\lambda\frac{p}{p+1}i_{X_a}d\omega.z.\psi-\lambda^2p\omega.e_a.z.\psi+2\lambda^2p(e_a\wedge\omega).z.\psi\nonumber\\
&&-\lambda\frac{p}{2(p+1)}d\omega.e_a.z.\psi\nonumber\\
&=&\lambda\frac{p}{p+1}i_{X_a}d\omega.z.\psi-\lambda^2p(e_a\wedge\omega-i_{X_a}\omega).z.\psi\nonumber\\
&&+2\lambda^2p(e_a\wedge\omega).z.\psi+\lambda\frac{p}{2(p+1)}(e_a\wedge d\omega-i_{X_a}d\omega).z.\psi\nonumber\\
&=&\lambda^2p(e_a\wedge\omega+i_{X_a}\omega).z.\psi+\lambda\frac{p}{2(p+1)}(e_a\wedge d\omega+i_{X_a}d\omega).z.\psi.\nonumber\\
\end{eqnarray}
From (22) and the operator $K_{\omega}$ in (101), we have
\begin{eqnarray}
\nabla_{X_a}K_{\omega}\psi&=&\lambda^2pe_a.\omega.z.\psi+\lambda\frac{p}{2(p+1)}e_a.d\omega.z.\psi\nonumber\\
&=&\lambda e_a.K_{\omega}\psi
\end{eqnarray}
which is the equation (99). So, $K_{\omega}$ defined in (98) with respect to a special even KY form $\omega$ is a symmetry operator of the Killing spinor equation in even dimensions.

On the other hand, if $\omega$ is a special odd CCKY $p$-form, then we have $d\omega=0$ and (98) turns into
\begin{eqnarray}
K_{\omega}\psi&=&(i_{X^a}\omega).z.\nabla_{X_a}\psi+\frac{p}{2(n-p+1)}\delta\omega.z.\psi\nonumber\\
&=&-\lambda(i_{X^a}\omega).e_a.z.\psi+\frac{p}{2(n-p+1)}\delta\omega.z.\psi
\end{eqnarray}
where we have used the Killing spinor equation (11) and the equality $z.e_a=-e_a.z$ in even dimensions. Moreover, since $\omega$ is an odd form, $i_{X^a}\omega$ is an even form and we have $i_{X^a}\omega.e_a=e_a\wedge i_{X^a}\omega=p\omega$. So, we obtain
\begin{eqnarray}
K_{\omega}\psi&=&-\lambda p\omega.z.\psi+\frac{p}{2(n-p+1)}\delta\omega.z.\psi.
\end{eqnarray}
To find the left hand side of (99), we can take the covariant derivative of (106) with respect to $X_a$ and this gives
\begin{eqnarray}
\nabla_{X_a}K_{\omega}\psi&=&-\lambda p\left(\nabla_{X_a}\omega.z.\psi+\omega.z.\nabla_{X_a}\psi\right)\nonumber\\
&&+\frac{p}{2(n-p+1)}\left(\nabla_{X_a}\delta\omega.z.\psi+\delta\omega.z.\nabla_{X_a}\psi\right)\nonumber\\
&=&\lambda\frac{p}{n-p+1}(e_a\wedge\delta\omega).z.\psi-\lambda^2p\omega.z.e_a.\psi+\frac{cp}{2}i_{X_a}\omega.z.\psi\nonumber\\
&&+\lambda\frac{p}{2(n-p+1)}\delta\omega.z.e_a.\psi
\end{eqnarray}
where we have used $\nabla_{X_a}z=0$ and the equations (11), (17) and (21). Since the odd CCKY $p$-form $\omega$ is generated by a Killing spinor with the Killing number $\lambda$, we have $c=-4\lambda^2$ from (32). By using the relation $z.e_a=-e_a.z$ and the identities $\omega.e_a=-e_a\wedge\omega+i_{X_a}\omega$ and $\delta\omega.e_a=e_a\wedge\delta\omega-i_{X_a}\delta\omega$ resulting from (23) since $\omega$ is an odd form, we obtain
\begin{eqnarray}
\nabla_{X_a}K_{\omega}\psi&=&\lambda\frac{p}{n-p+1}(e_a\wedge\delta\omega).z.\psi+\lambda^2p\omega.e_a.z.\psi-2\lambda^2pi_{X_a}\omega.z.\psi\nonumber\\
&&-\lambda\frac{p}{2(n-p+1)}\delta\omega.e_a.z.\psi\nonumber\\
&=&\lambda\frac{p}{n-p+1}(e_a\wedge\delta\omega).z.\psi-\lambda^2p(e_a\wedge\omega-i_{X_a}\omega).z.\psi\nonumber\\
&&-2\lambda^2pi_{X_a}\omega.z.\psi-\lambda\frac{p}{2(n-p+1)}(e_a\wedge\delta\omega-i_{X_a}\delta\omega).z.\psi\nonumber\\
&=&-\lambda^2p(e_a\wedge\omega+i_{X_a}\omega).z.\psi\nonumber\\
&&+\lambda\frac{p}{2(n-p+1)}(e_a\wedge\delta\omega+i_{X_a}\delta\omega).z.\psi.
\end{eqnarray}
From (22) and the operator $K_{\omega}$ in (106), we have
\begin{eqnarray}
\nabla_{X_a}K_{\omega}\psi&=&-\lambda^2pe_a.\omega.z.\psi+\lambda\frac{p}{2(n-p+1)}e_a.\delta\omega.z.\psi\nonumber\\
&=&\lambda e_a.K_{\omega}\psi
\end{eqnarray}
which is the equation (99). So, $K_{\omega}$ defined in (98) with respect to a special odd CCKY form $\omega$ is also a symmetry operator of the Killing spinor equation in even dimensions.
\end{proof}

In odd dimensions, the volume form $z$ is at the center of the Clifford algebra and hence commutes with all of the elements. So, we have $z.e_a=e_a.z$. However, in that case the operator defined in (98) with respect to special even KY and special odd CCKY forms does not satisfy the symmetry operator condition (99) for the Killing spinor equation (11). So, there is no companion of Theorem 4 in odd dimensions. This fact resembles the non-existence of the Lie algebra structure of special even KY and special odd CCKY forms in odd dimensions discussed in Section 4 and we will see that they are related to each other.

\section{Generalized Symmetry Superalgebras}

We will see in this section that the structures considered in the previous sections such as the bilinears of Killing spinors, the Lie algebra structure of special KY and special CCKY forms and the symmetry operators of Killing spinors constructed out of them together constitute a superalgebra structure.
\begin{definition}
A superalgebra $\mathfrak{g}$ is a $\mathbb{Z}_2$-graded algebra that consists of a direct sum of two components $\mathfrak{g}=\mathfrak{g}_0\oplus\mathfrak{g}_1$ on which a bilinear operation
\begin{equation}
[\,,\,]:\mathfrak{g}_i\times\mathfrak{g}_j\rightarrow\mathfrak{g}_{i+j}
\end{equation}
is defined. Here $i,j=0,1(\textrm{mod }2)$ and $[\,,\,]$ satisfies the condition
\begin{equation}
[a,b]=-(-1)^{|a||b|}[b,a]
\end{equation}
for $a,b\in\mathfrak{g}$ and $|\,|$ denotes the degree of the element which is 0 or 1 depending on that it is in $\mathfrak{g}_0$ or $\mathfrak{g}_1$. The first component $\mathfrak{g}_0$ is called the even part and the second component $\mathfrak{g}_1$ is called the odd part of the superalgebra. Moreover, if $[\,,\,]$ satisfies the following graded Jacobi identity
\begin{equation}
[a,[b,c]]=[[a,b],c]+(-1)^{|a||b|}[b,[a,c]]
\end{equation}
then $\mathfrak{g}$ is called a Lie superalgebra.
\end{definition}

By considering the Killing spinors and Killing vector fields corresponding to the Dirac currents of Killing spinors on a manifold $M$, a superalgebra structure called a symmetry superalgebra can be defined.

\begin{definition}
A superalgebra $\mathfrak{k}=\mathfrak{k}_0\oplus\mathfrak{k}_1$ on $M$ is called a symmetry superalgebra if its even part $\mathfrak{k}_0$ consists of the Lie algebra of the Killing vector fields on $M$ and the odd part $\mathfrak{k}_1$ corresponds to the space of Killing spinors on $M$. The bilinear operation is defined as follows. The even-even bracket corresponds to the Lie bracket of vector fields defined in (35)
\begin{equation}
[\,,\,]:\mathfrak{k}_0\times\mathfrak{k}_0\rightarrow\mathfrak{k}_0
\end{equation}
The even-odd bracket is the Lie derivative on spinor fields  with respect to Killing vector fields defined in (82)
\begin{equation}
\pounds:\mathfrak{k}_0\times\mathfrak{k}_1\rightarrow\mathfrak{k}_1
\end{equation}
and the odd-odd bracket is the Dirac current of a spinor defined in (7)
\begin{equation}
(\,\,)_1:\mathfrak{k}_1\times\mathfrak{k}_1\rightarrow\mathfrak{k}_0
\end{equation}
These brackets satisfy the property (111).
\end{definition}
The Jacobi identities of the symmetry superalgebra correspond to the properties of the Lie bracket of vector fields and the Lie derivative on spinors. The even-even-even Jacobi identity is automatically satisfied since it is the Jacobi identity for the Lie bracket of vector fields
\begin{equation}
[K_1,[K_2,K_3]]+[K_2,[K_3,K_1]]+[K_3,[K_1,K_2]]=0
\end{equation}
where $K_1, K_2, K_3 \in\mathfrak{k}_0$. The even-even-odd and even-odd-odd Jacobi identites correspond to the following properties of the Lie derivative on spinors and are automatically satisfied
\begin{eqnarray}
[\pounds_{K_1},\pounds_{K_2}]\psi&=&\pounds_{[K_1,K_2]}\psi\\
{\mathcal{L}}_K(\psi\overline{\phi})&=&(\pounds_K\psi)\overline{\phi}+\psi(\overline{\pounds_K\phi})
\end{eqnarray}
where $K, K_1, K_2\in\mathfrak{k}_0$, $\psi, \phi\in\mathfrak{k}_1$ and ${\mathcal{L}}$ is the Lie derivative on Clifford forms. However, the odd-odd-odd Jacobi identity for $\psi\in\mathfrak{k}_1$
\begin{equation}
\pounds_{V_{\psi}}\psi=0
\end{equation}
is not automatically satisfied and in the cases that it is satisfied the symmetry superalgebra is a Lie superalgebra.

The construction of the symmetry superalgebra can be extended to include the odd degree KY forms \cite{Ertem1,Ertem2}.

\begin{definition}
A superalgebra $\overline{\mathfrak{k}}=\overline{\mathfrak{k}}_0\oplus\overline{\mathfrak{k}}_1$ on $M$ is called an extended symmetry superalgebra if its even part $\overline{\mathfrak{k}}_0$ consists of the Lie algebra of the odd degree KY forms on $M$ and the odd part $\mathfrak{k}_1$ corresponds to the space of Killing spinors on $M$. The bilinear operation is defined as follows. The even-even bracket corresponds to the SN bracket of KY forms defined in (39)
\begin{equation}
[\,,\,]_{SN}:\overline{\mathfrak{k}}_0\times\overline{\mathfrak{k}}_0\rightarrow\overline{\mathfrak{k}}_0
\end{equation}
The even-odd bracket is the symmetry operators of Killing spinors constructed out of odd KY forms defined in (86)
\begin{equation}
L:\overline{\mathfrak{k}}_0\times\overline{\mathfrak{k}}_1\rightarrow\overline{\mathfrak{k}}_1
\end{equation}
and the odd-odd bracket is the $p$-form Dirac currents of Killing spinors defined in (8)
\begin{equation}
(\,\,)_p:\overline{\mathfrak{k}}_1\times\overline{\mathfrak{k}}_1\rightarrow\overline{\mathfrak{k}}_0
\end{equation}
These brackets satisfy the property (111).
\end{definition}
Although the even-even-even Jacobi identity of the extended symmetry superalgebra is automatically satisfied since it corresponds to the Jacobi identity for the SN bracket of odd degree forms, other Jacobi identities are not automatically satisfied and in general the extended symmetry superalgebra is not a Lie superalgebra.

We generalize the construction of symmetry and extended symmetry superalgebras to include all special KY and special CCKY forms generated by Killing spinors on the manifold $M$. By using Proposition 1, Theorem 1 and Theorem 3, the generalized symmetry superalgebras are constructed as in the following
\begin{theorem}
On an $n$-dimensional manifold $M$ admitting Killing spinors, if the bilinears of Killing spinors correspond to the special odd KY and special even CCKY forms, then the space of Killing spinors and special odd KY and special even CCKY forms generated by them constitute a superalgebra structure which we denote by $\mathfrak{K}=\mathfrak{K}_0\oplus\mathfrak{K}_1$. The even part $\mathfrak{K}_0$ of the superalgebra is the Lie algebra of the special odd KY and special even CCKY forms with respect to the CKY bracket and the odd part $\mathfrak{K}_1$ is the space of Killing spinors. This superalgebra is called the generalized symmetry superalgebra and the bilinear operation is defined as follows. The even-even bracket is the CKY bracket of the special odd KY and special even CCKY forms defined in (48)
\begin{equation}
[\,,\,]_{CKY}:\mathfrak{K}_0\otimes\mathfrak{K}_0\rightarrow\mathfrak{K}_0
\end{equation}
The even-odd bracket is the symmetry operators of Killing spinors generated from the special odd KY and special even CCKY forms defined in (89)
\begin{equation}
L:\mathfrak{K}_0\otimes\mathfrak{K}_1\rightarrow\mathfrak{K}_1
\end{equation}
The odd-odd bracket is the $p$-form bilinears of Killing spinors defined in (8)
\begin{equation}
(\,\,)_p:\mathfrak{K}_1\otimes\mathfrak{K}_1\rightarrow\mathfrak{K}_0
\end{equation}
These brackets satisfy the condition (111) and hence constitute a superalgebra structure.
\end{theorem}
\begin{proof}
As we proved in Theorem 1, the CKY bracket $[\,,\,]_{CKY}$ is a Lie bracket for special odd KY and special even CCKY forms and $\mathfrak{K}_0$ is a Lie algebra with respect to the bracket in (123). We have also proved in Theorem 3 that the operator (89) constructed out of special odd KY forms and special even CCKY forms are symmetry operators for Killing spinors. Hence, the bracket given in (124) is well-defined for $\mathfrak{K}_0$ and $\mathfrak{K}_1$. Moreover, we have also shown in Proposition 1 that the symmetric combination of $p$-form bilinears of two Killing spinors correspond to special odd KY and special even CCKY forms in relevant situations. So, the bracket defined in (125) is also well-defined for $\mathfrak{K}_0$ and $\mathfrak{K_1}$. As a consequence, $\mathfrak{K}=\mathfrak{K}_0\oplus\mathfrak{K}_1$ constitute a consistent superalgebra structure.
\end{proof}

On the other hand, we can also construct a superalgebra in even dimensions for the case of special even KY forms and special odd CCKY forms generated by the Killing spinors.
\begin{theorem}
On an even $n$-dimensional manifold $M$ admitting Killing spinors, if the bilinears of Killing spinors correspond to the special even KY and special odd CCKY forms, then the space of Killing spinors and special even KY and special odd CCKY forms generated by them constitute a superalgebra structure which we denote by $\overline{\mathfrak{K}}=\overline{\mathfrak{K}}_0\oplus\overline{\mathfrak{K}}_1$. The even part $\overline{\mathfrak{K}}_0$ of the superalgebra is the Lie algebra of the special even KY and special odd CCKY forms with respect to the CKY bracket and the odd part $\overline{\mathfrak{K}}_1$ is the space of Killing spinors. This superalgebra is called the generalized symmetry superalgebra and the bilinear operation is defined as follows. The even-even bracket is the Hodge star of the CKY bracket of the special even KY and special odd CCKY forms defined in Theorem 2
\begin{equation}
*[\,,\,]_{CKY}:\overline{\mathfrak{K}}_0\otimes\overline{\mathfrak{K}}_0\rightarrow\overline{\mathfrak{K}}_0
\end{equation}
The even-odd bracket is the symmetry operators of Killing spinors generated from the special even KY and special odd CCKY forms defined in (98)
\begin{equation}
K:\overline{\mathfrak{K}}_0\otimes\overline{\mathfrak{K}}_1\rightarrow\overline{\mathfrak{K}}_1
\end{equation}
The odd-odd bracket is the $p$-form bilinears of Killing spinors defined in (8)
\begin{equation}
(\,\,)_p:\overline{\mathfrak{K}}_1\otimes\overline{\mathfrak{K}}_1\rightarrow\overline{\mathfrak{K}}_0
\end{equation}
These brackets satisfy the condition (111) and hence constitute a superalgebra structure in even dimensions.
\end{theorem}
\begin{proof}
As we have shown in Theorem 2, the Hodge star of the CKY bracket $*[\,,\,]_{CKY}$ is a Lie bracket for special even KY and special odd CCKY forms in even dimensions and $\overline{\mathfrak{K}}_0$ is a Lie algebra with respect to the bracket in (126). We have also shown in Theorem 4 that the operator (98) constructed out of special even KY forms and special odd CCKY forms are symmetry operators for Killing spinors in even dimensions. Hence, the bracket given in (127) is well-defined for $\overline{\mathfrak{K}}_0$ and $\overline{\mathfrak{K}}_1$. Moreover, we have also shown in Proposition 1 that the symmetric combination of $p$-form bilinears of two Killing spinors correspond to special even KY and special odd CCKY forms in relevant situations. So, the bracket defined in (128) is also well-defined for $\overline{\mathfrak{K}}_0$ and $\overline{\mathfrak{K}}_1$. As a consequence, $\overline{\mathfrak{K}}=\overline{\mathfrak{K}}_0\oplus\overline{\mathfrak{K}}_1$ constitute a consistent superalgebra structure in even dimensions.
\end{proof}

Note that there is no superalgebra structure for special even KY and special odd CCKY forms generated by Killing spinors in odd dimensions. The reason for this is that the bracket $*[\,,\,]_{CKY}$ is not a Lie bracket and the operator $K$ is not a symmetry operator in odd dimensions.

The even-even-even Jacobi identities of the superalgebra structures constructed in Theorems 5 and 6 are automatically satisfied because of the properties of $[\,,\,]_{CKY}$ and $*[\,,\,]_{CKY}$ brackets. However, other Jacobi identities are not automatically satisfied and the superalgebras does not correspond to Lie superalgebras in general.

\section{Examples}

As examples for the generalized symmetry superalgebras constructed in Section 6, we consider special types of 6 and 7-dimensional compact Riemannian manifolds called nearly K\"{a}hler and weak $G_2$ manifolds \cite{Joyce,Besse}.

\subsection{Weak $G_2$ Manifolds}

\begin{definition}
A 7-dimensional compact Riemannian manifold $M$ is called a weak $G_2$ manifold, if its metric cone has holonomy contained in $Spin(7)$. $M$ admits a Killing spinor $\psi$ satisfying the Killing spinor equation (11)
\[
\nabla_{X}\psi=\lambda\widetilde{X}.\psi
\]
for any vector field $X$ and the Killing number $\lambda$. This implies that $M$ is an Einstein manifold \cite{Friedrich Kath Moroianu Semmelmann}.
\end{definition}

For a 7-dimensional Riemannian manifold $M$, the Clifford algebra defined on the tangent bundle $TM$ is isomorphic to $Cl_{7,0}\cong\mathbb{C}(8)$ where $\mathbb{C}(8)$ is the space of $8\times 8$ dimensional complex matrices. The even subalgebra of a real Clifford algebra can be found from the one lower dimensional Clifford algebra by the relation $Cl^0_{p+1,q}\cong Cl_{q,p}$ for any $p$ and $q$. So, the even subalgebra for the 7-dimensional Riemannian case is $Cl^0_{7,0}\cong Cl_{0,6}\cong\mathbb{R}(8)$. The space carrying the irreducible representations of the even subalgebra corresponds to the spinor space $S$ and in our case we have $S\cong\mathbb{R}^8$. Then, spinors are real and correspond to Majorana spinors which means that the Killing spinor $\psi$ is a real Killing spinor and the Killing number $\lambda$ is a real number. The spin invariant inner product defined on $S$ is a $\mathbb{R}$-symmetric inner product with $\xi\eta$ involution. So, we have ${\mathcal{J}}=\xi\eta$, $j=\textrm{Id}$ and the following properties for the inner product
\begin{eqnarray}
(\psi,\phi)&=&(\phi,\psi)\\
(\psi,\alpha.\phi)&=&(\alpha^{\xi\eta}.\psi,\phi)
\end{eqnarray}
where $\psi,\phi\in S$ and $\alpha\in Cl(M)$.

Now, we can calculate the $p$-form Dirac currents of the Killing spinor $\psi$;
\begin{eqnarray}
(\psi\overline{\psi})_0&=&(\psi,\psi)\neq 0\nonumber\\
(\psi\overline{\psi})_1&=&(\psi,e_a.\psi)e^a=(e_a^{\xi\eta}.\psi,\psi)e^a\nonumber\\
&=&-(e_a.\psi,\psi)e^a=-(\psi,e_a.\psi)e^a=0\nonumber\\
(\psi\overline{\psi})_2&=&(\psi,e_{ba}.\psi)e^{ab}=(e_{ba}^{\xi\eta}.\psi,\psi)e^{ab}\nonumber\\
&=&-(e_{ba}.\psi,\psi)e^{ab}=-(\psi,e_{ba}.\psi)e^{ab}=0\nonumber\\
(\psi\overline{\psi})_3&=&(\psi,e_{cba}.\psi)e^{abc}=(e_{cba}^{\xi\eta}.\psi,\psi)e^{abc}\nonumber\\
&=&(e_{cba}.\psi,\psi)e^{abc}=(\psi,e_{cba}.\psi)e^{abc}\neq 0\nonumber
\end{eqnarray}
and by similar reasoning, we also have
\begin{eqnarray}
(\psi\overline{\psi})_4&\neq&0\nonumber\\
(\psi\overline{\psi})_5&=&0\nonumber\\
(\psi\overline{\psi})_6&=&0\nonumber\\
(\psi\overline{\psi})_7&\neq&0\nonumber.
\end{eqnarray}
From Proposition 1, one can conclude that we have two non-zero special odd KY forms $(\psi\overline{\psi})_3$ and $(\psi\overline{\psi})_7$ and two non-zero special even CCKY forms $(\psi\overline{\psi})_0$ and $(\psi\overline{\psi})_4$. So, they are Hodge duals of each other; $(\psi\overline{\psi})_3=*(\psi\overline{\psi})_4$ and $(\psi\overline{\psi})_0=*(\psi\overline{\psi})_7$. Moreover, since we have $d(\psi\overline{\psi})_7=0$ and $\delta(\psi\overline{\psi})_0=0$, the bilinears $(\psi\overline{\psi})_0$ and $(\psi\overline{\psi})_7$ correspond to parallel forms and we can choose as $(\psi\overline{\psi})_0=1$ and $(\psi\overline{\psi})_7=z$ where $z$ is the volume form. From the equations (23) and (25), $(\psi\overline{\psi})_3$ and $(\psi\overline{\psi})_4$ satisfy the following equalities
\begin{eqnarray}
d(\psi\overline{\psi})_3&=&8\lambda(\psi\overline{\psi})_4\\
\delta(\psi\overline{\psi})_4&=&-8\lambda(\psi\overline{\psi})_3.
\end{eqnarray}

We can check the superalgebra structure formed by the Killing spinor $\psi$ and the special odd KY and special even CCKY forms generated by $\psi$. From the definition of the CKY bracket in (48), one can check that the only non-vanishing brackets of the $p$-form Dirac currents are
\begin{eqnarray}
{[(\psi\overline{\psi})_0,(\psi\overline{\psi})_{4}]}_{CKY}&=&-2\lambda(\psi\overline{\psi})_3\nonumber\\
{[(\psi\overline{\psi})_4,(\psi\overline{\psi})_{4}]}_{CKY}&=&k(\psi\overline{\psi})_7
\end{eqnarray}
where $k$ is a constant and all other brackets vanish.

To construct the symmetry operators defined in (89), we first consider the algebraic relations between spinor bilinears. From (6), we have
\[
\psi\overline{\psi}=(\psi\overline{\psi})_0+(\psi\overline{\psi})_3+(\psi\overline{\psi})_4+(\psi\overline{\psi})_7
\]
and we can write
\begin{eqnarray}
(\psi\overline{\psi})_3.\psi+(\psi\overline{\psi})_4.\psi+(\psi\overline{\psi})_7.\psi&=&(\psi\overline{\psi})\psi-(\psi\overline{\psi})_0.\psi\nonumber\\
&=&(\psi,\psi)\psi-(\psi\overline{\psi})_0.\psi\nonumber\\
&=&(\psi\overline{\psi})_0.\psi-(\psi\overline{\psi})_0.\psi\nonumber\\
&=&0
\end{eqnarray}
where we have used (5) and the definition $(\psi\overline{\psi})_0=( \psi,\psi)$. In a 7-dimensional Riemannian manifold, the volume form has the property $z^2=-1$. This implies that $z^2.\psi=-\psi$ and hence $z.\psi=\pm i\psi$. So, we have $(\psi\overline{\psi})_7.\psi=z.\psi=\pm i\psi$ and from (134), we can write
\begin{equation}
(\psi\overline{\psi})_3.\psi+(\psi\overline{\psi})_4.\psi=-(\psi\overline{\psi})_7.\psi=\mp i\psi.
\end{equation}
This means that we have the following identities
\begin{eqnarray}
(\psi\overline{\psi})_3.\psi&=&\frac{\mp i}{1\mp i}\psi\\
(\psi\overline{\psi})_4.\psi&=&-\frac{1}{1\mp i}\psi
\end{eqnarray}
where we have used that $(\psi\overline{\psi})_4.\psi=*(\psi\overline{\psi})_3.\psi=(\psi\overline{\psi})_3^{\xi}.z.\psi=-(\psi\overline{\psi})_3.z.\psi$. By using these identities, we can construct the symmetry operators defined in (89) as follows
\begin{eqnarray}
L_{(\psi\overline{\psi})_0}\psi&=&0\nonumber\\
L_{(\psi\overline{\psi})_3}\psi&=&\mp 3i\lambda\psi\nonumber\\
L_{(\psi\overline{\psi})_4}\psi&=&-4\lambda\psi\\
L_{(\psi\overline{\psi})_7}\psi&=&\pm 7i\lambda\psi\nonumber
\end{eqnarray}
where we have used the identities $d(\psi\overline{\psi})_3=8\lambda(\psi\overline{\psi})_4$ and $\delta(\psi\overline{\psi})_4=-8\lambda(\psi\overline{\psi})_3$ given in (27) and (29) with the closedness and co-closedness properties of CCKY and KY forms respectively.

So, we have constructed all the brackets of the superalgebra defined in Theorem 5 for a 7-dimensional weak $G_2$ manifold. To see that whether it is also a Lie superalgebra, we need to check the Jacobi identities. The first Jacobi identity is the the Jacobi identity for the CKY bracket which is automatically satisfied. The second Jacobi identity is the vanishing of the commutators of symmetry operators given by
\begin{equation}
[L_{(\psi\overline{\psi})_i},L_{(\psi\overline{\psi})_j}]=0
\end{equation}
and it can be seen from (138) that it is satisfied for all cases of $i,j=0,3,4,7$. However, the third and fourth Jacobi identities given by
\begin{eqnarray}
[(\psi\overline{\psi})_i,(\psi\overline{\psi})]_{CKY}&=&(L_{(\psi\overline{\psi})_i}\psi)\overline{\psi}+\psi(\overline{L_{(\psi\overline{\psi})_i}\psi})\\
L_{(\psi\overline{\psi})}\psi&=&0
\end{eqnarray}
are not satisfied for all cases. For example, we have
\begin{equation}
L_{(\psi\overline{\psi})}\psi=-4\lambda(1\mp i)\psi
\end{equation}
which is not zero. So, the generalized symmetry superalgebra of a weak $G_2$ manifold is a superalgebra but not a Lie superalgebra.

\subsection{Nearly K\"{a}hler Manifolds}

\begin{definition}
A Riemannian manifold $M$ with metric $g$ and almost complex structure $J$ is called a nearly K\"{a}hler manifold if $J$ satisfies the condition
\begin{equation}
i_X\nabla_XJ=0
\end{equation}
for every vector field $X\in TM$. The metric cone of a 6-dimensional nearly K\"{a}hler manifold $M$ has holonomy in $G_2$ and $M$ admits two Killing spinors with different chirality \cite{Joyce,Besse}.
\end{definition}

In a 6-dimensional Riemannian manifold $M$, the Clifford algebra is isomorphic to $Cl_{6,0}\cong\mathbb{H}(4)$ where $\mathbb{H}(4)$ is the space of $4\times 4$ dimensional quaternionic matrices. The even subalgebra corresponds to $Cl^0_{6,0}\cong Cl_{0,5}\cong\mathbb{C}(4)$. There are two different irreducible representations of the even subalgebra and the spinor space $S=S_+\oplus S_-$ is equivalent to $S_+\oplus S_-\cong\mathbb{C}^4\oplus\mathbb{C}^4$. Then, spinors are chiral complex and correspond to Dirac-Weyl spinors. The Killing spinors $\psi^+\in S_+$ and $\psi^-\in S_-$ on $M$ are real Killing spinors and the Killing numbers $\lambda_+$ and $\lambda_-$ are both real numbers. The spin invariant inner product defined on $S$ is a $\mathbb{C}$-skew inner product with $\xi$ involution. So, we have ${\mathcal{J}}=\xi$, $j=\textrm{Id}$ and the following properties for the inner product
\begin{eqnarray}
(\psi,\phi)&=&-(\phi,\psi)\\
(\psi,\alpha.\phi)&=&(\alpha^{\xi}.\psi,\phi)
\end{eqnarray}
where $\psi,\phi\in S$ and $\alpha\in Cl(M)$.

We can find the $p$-form Dirac currents of the Killing spinor $\psi^+$ as follows
\begin{eqnarray}
(\psi^+\overline{\psi^+})_0&=&-(\psi^+,\psi^+)=0\nonumber\\
(\psi^+\overline{\psi^+})_1&=&(\psi^+,e_a.\psi^+)e^a=(e_a^{\xi}.\psi^+,\psi^+)e^a\nonumber\\
&=&(e_a.\psi^+,\psi^+)e^a=-(\psi^+,e_a.\psi^+)e^a=0\nonumber\\
(\psi^+\overline{\psi^+})_2&=&(\psi^+,e_{ba}.\psi^+)e^{ab}=(e_{ba}^{\xi}.\psi^+,\psi^+)e^{ab}\nonumber\\
&=&-(e_{ba}.\psi^+,\psi^+)e^{ab}=(\psi^+,e_{ba}.\psi^+)e^{ab}\neq 0
\end{eqnarray}
and by similar reasoning, we also have
\begin{eqnarray}
(\psi^+\overline{\psi^+})_3&\neq&0\nonumber\\
(\psi^+\overline{\psi^+})_4&=&0\nonumber\\
(\psi^+\overline{\psi^+})_5&=&0\nonumber\\
(\psi^+\overline{\psi^+})_6&\neq&0\nonumber.
\end{eqnarray}
The same equalities are also true for the other Killing spinor $\psi^-$. By considering the cases for ${\mathcal{J}}=\xi$, $\lambda$ is real and $j=\textrm{Id}$ in Proposition 1, we reach to the fact that the even degree bilinears $(\psi^+\overline{\psi^+})_2$, $(\psi^+\overline{\psi^+})_6$, $(\psi^-\overline{\psi^-})_2$ and $(\psi^-\overline{\psi^-})_6$ are special even KY forms and the odd degree bilinears $(\psi^+\overline{\psi^+})_3$ and $(\psi^-\overline{\psi^-})_3$ are special odd CCKY forms. Moreover, $(\psi^+\overline{\psi^+})_6$ and $(\psi^-\overline{\psi^-})_6$ correspond to the volume form $z$ and we have the identities
\begin{eqnarray}
\delta(\psi\overline{\psi})_2&=&0\\
d(\psi\overline{\psi})_2&=&6\lambda(\psi\overline{\psi})_3\\
\delta(\psi\overline{\psi})_3&=&-8\lambda(\psi\overline{\psi})_2\\
d(\psi\overline{\psi})_3&=&0
\end{eqnarray}
for both $\psi^+$ and $\psi^-$.

From Theorem 6, the bracket of the even part of the superalgebra is $*[\,,\,]_{CKY}$ and the only non-zero bracket in the even part is
\begin{equation}
*[(\psi^+\overline{\psi^+})_2,(\psi^-\overline{\psi^-})_2]_{CKY}=a(\psi^+\overline{\psi^+})_3+b(\psi^-\overline{\psi^-})_3
\end{equation}
where $a$ and $b$ are constants.

From the following algebraic relations
\begin{eqnarray}
(\psi^+\overline{\psi^+})_2.\psi^++(\psi^+\overline{\psi^+})_3.\psi^++(\psi^+\overline{\psi^+})_6.\psi^+=0\nonumber\\
(\psi^+\overline{\psi^+})_2.\psi^-+(\psi^+\overline{\psi^+})_3.\psi^-+(\psi^+\overline{\psi^+})_6.\psi^-=0\nonumber
\end{eqnarray}
and the same formulas for the interchange of $\psi^+$ and $\psi^-$, we can find the symmetry operators given in (98) as follows
\begin{eqnarray}
K_{(\psi^+\overline{\psi^+})_2}\psi^+&=&-2\lambda\psi^+\nonumber\\
K_{(\psi^+\overline{\psi^+})_2}\psi^-&=&-2\lambda\psi^-\nonumber\\
K_{(\psi^+\overline{\psi^+})_3}\psi^+&=&3\lambda\psi^+\nonumber\\
K_{(\psi^+\overline{\psi^+})_3}\psi^-&=&-2\lambda\psi^-\nonumber\\
K_{(\psi^+\overline{\psi^+})_6}\psi^+&=&-6\lambda\psi^+\nonumber\\
K_{(\psi^+\overline{\psi^+})_6}\psi^-&=&-6\lambda\psi^-
\end{eqnarray}
and the same equalities are true for the symmetry operators constructed from $\psi^-$.

To see that whether it is also a Lie superalgebra, we need to check the Jacobi identities. The first Jacobi identity is the the Jacobi identity for the $*[\,,\,]_{CKY}$ bracket which is automatically satisfied. The second Jacobi identity is the vanishing of the commutators of symmetry operators given by
\begin{equation}
[K_{(\psi\overline{\psi})_i},K_{(\psi\overline{\psi})_j}]=0
\end{equation}
for $\psi^+$ and $\psi^-$ and it can be seen from (152) that it is satisfied for all cases of $i,j=2,3,6$. However, the third and fourth Jacobi identities are not satisfied for all cases. So, the generalized symmetry superalgebra of a nearly K\"{a}hler manifold is a superalgebra but not a Lie superalgebra.

\section{Conclusion}

We generalize the symmetry superalgebras that correspond to geometric invariants of manifolds with isometries to include all the hidden symmetries of the manifold generated by geometric Killing spinors. This defines a more complete construction of the superalgebra structure of symmetries of the manifold. This also gives way to construct generalizations of the Lie derivative on spinor fields as symmetry operators of geometric Killing spinors and the construction of the Lie algebra structure of special KY and special CCKY forms.

Besides the symmetry superalgebras of isometries, one can also construct conformal superalgebras from conformal symmetries and twistor spinors \cite{de Medeiros Hollands,Ertem5}. Twistor spinors correspond to supersymmetry generators of superconformal field theories and the construction of conformal superalgebras are related to the classification of the superconformal backgrounds in these theories. The methods described in the paper can also be used to obtain generalized conformal superalgebras and generalized gauged conformal superalgebras whose supersymmetry generators correspond to gauged twistor spinors \cite{Ertem6,Ertem7}. On the other hand, Killing superalgebras of supergravity backgrounds in non-constant curvature backgrounds can also lead to more general geometric invariants by investigating the generalizations with using the similar methods described in the paper and in \cite{Acik Ertem2}. This can give new perspectives to the classification problem of supergravity backgrounds in all dimensions and supergravity theories.


\end{document}